\documentclass[aps,onecolumn,a4paper,11pt,tightenlines,superscriptaddress,longbibliography,nofootinbib,notitlepage]{revtex4-1}

\usepackage[utf8]{inputenc}
\usepackage[T1]{fontenc}
\usepackage{lmodern}
\usepackage{amsmath,amssymb,amsfonts,mathtools,array}
\usepackage{amsthm}
\usepackage{microtype}
\usepackage{booktabs}
\usepackage{color}
\usepackage{bm}
\usepackage{graphicx}
\usepackage{placeins}
\DeclareMathOperator{\Tr}{Tr}
\newcommand{\ket}[1]{\lvert#1\rangle}
\newcommand{\bra}[1]{\langle#1\rvert}
\usepackage{tikz}
\usetikzlibrary{arrows.meta,positioning,calc,fit,shapes.geometric}
\usepackage{url}
\usepackage[hidelinks]{hyperref}
\hypersetup{pdftitle={From Symmetry to Secrecy: Covariant Classical--Quantum Wiretap Channels},pdfauthor={Masahito Hayashi and Farzin Salek}}
\newtheorem{theorem}{Theorem}[section]
\newtheorem{definition}[theorem]{Definition}
\newtheorem{lemma}[theorem]{Lemma}
\newtheorem{proposition}[theorem]{Proposition}
\newtheorem{corollary}[theorem]{Corollary}
\newtheorem{remark}[theorem]{Remark}

\DeclareMathOperator*{\argmax}{arg\,max}

\DeclareMathOperator{\supp}{supp}
\DeclareMathOperator{\Stab}{Stab}
\DeclareMathOperator{\atanh}{atanh}
\DeclareMathOperator{\Bern}{Bern}
\DeclareMathOperator{\Rep}{Rep}

\newcommand{\sM}{\mathsf{M}}
\newcommand{\sL}{\mathsf{L}}
\newcommand{\sK}{\mathsf{K}}
\newcommand{\sJ}{\mathsf{J}}

\newcommand{\bbZ}{\mathbb{Z}}

\newcommand{\bbC}{\mathbb{C}}
\newcommand{\bbE}{\mathbb{E}}

\newcommand{\calB}{\mathcal{B}}
\newcommand{\calD}{\mathcal{D}}
\newcommand{\calH}{\mathcal{H}}
\newcommand{\calN}{\mathcal{N}}
\newcommand{\calP}{\mathcal{P}}
\newcommand{\calS}{\mathcal{S}}

\newcommand{\calX}{\mathcal{X}}

\newcommand{\Pu}{P_{\mathrm{u}}}
\newcommand{\DeltaBE}{\Delta}
\newcommand{\rhoBbar}{\overline{\rho}_{B}}
\newcommand{\rhoEbar}{\overline{\rho}_{E}}

\newcommand{\Label}{\label}
\newcommand{\ach}{\mathrm{ach}}
\newcommand{\rep}{\mathrm{rep}}
\newcommand{\pair}{\mathrm{pair}}

\begin{document}

\title{From Symmetry to Secrecy:\\
Covariant Classical--Quantum Wiretap Channels}

\author{Masahito Hayashi}
\email{hmasahito@cuhk.edu.cn}
\affiliation{School of Data Science, The Chinese University of Hong Kong, Shenzhen, Longgang District, Shenzhen 518172, China}
\affiliation{International Quantum Academy, Futian District, Shenzhen 518048, China}
\affiliation{Graduate School of Mathematics, Nagoya University, Nagoya 464-8602, Japan}
\author{Farzin Salek}
\email{farzin.salek@gmail.com}
\affiliation{Dahlem Center for Complex Quantum Systems, Freie Universit\"at Berlin, Berlin, Germany}
\affiliation{Perimeter Institute for Theoretical Physics, Waterloo, ON N2L 2Y5, Canada}
\affiliation{Institute for Quantum Computing, University of Waterloo, Waterloo, ON N2L 3G1, Canada}

\begin{abstract}
For classical--quantum wiretap channels with a transitive covariant input
action, both output ensembles admit a description as unitary group orbits
of fixed seed states. We use this structure to construct finite-block
preprocessing schemes whose private information is an entropy difference
evaluated at two distributions related by group convolution.
Applying the construction to product groups allows correlations
across any fixed number of channel uses. We give a sufficient condition for
the associated uniform-input bound to equal the single-use private
information. Constant-composition channel codes and universal$_2$ hashing
achieve the resulting rates, with trace-distance leakage controlled by
sandwiched R\'enyi information. For uniform outer randomization, the
representation structure reduces the comparison-state optimization to Eve's
invariant states. For the
hybrid binary channel considered by Tikku, Berta and Renes, numerical
optimization at three and four uses gives rates above our numerical estimate
of the optimized noisy-repetition benchmark wherever that estimate is positive
on the tested grid. The improvement comes from noise distributions beyond the family of
independent physical flips and a symmetric logical flip. We also prove
nonadditivity at an explicit channel parameter using a two-use encoder, an
analytic single-use converse, and certified entropy bounds. Finally, we
relate the hybrid model to binary phase-shift keying with a fixed receiver
measurement.
\end{abstract}

\maketitle
\vspace{-2pt}

\section{Introduction}\Label{sec:introduction}

A classical--quantum (cq) wiretap channel takes a classical symbol $x$ and
delivers states $W_x^B$ and $W_x^E$ to a legitimate receiver, Bob, and an
eavesdropper, Eve. Private communication is governed by the difference
$I(V;B)-I(V;E)$, where the auxiliary variable $V$ labels the signals Alice
wishes to distinguish and $P_{X|V}$ describes her randomized encoder
\cite{CK,DevetakPrivate}. For classical wiretap channels, optimizing this
quantity for one use gives the capacity. For cq wiretap channels,
correlations across several uses can increase the rate per use. Tikku,
Berta and Renes demonstrated this even when only one receiver has a quantum
output \cite{TBR}. In their example with quantum Eve, the construction uses
noisy repetition encoding, building on structured codes for quantum key
distribution \cite{SRS}.

Randomized encoding and auxiliary variables also organize broadcast
coding. In classical models, evaluating the resulting bounds leads to
cardinality and concave-envelope questions
\cite{GA12,AGN19}. For quantum outputs, superposition and Marton coding
have been developed for two receivers \cite{SavovWilde15}, and quantum
nonunique decoding gives further constructions for three receivers
\cite{SHH26}. Confidential broadcast coding brings the encoder's local
randomness into the rate tradeoff \cite{SHF20}; recent work treats two
mutually confidential messages through likelihood encoding and bipartite
cq resolvability \cite{BGSS26}. The related problem of cq channel
resolvability, including worst-input and fixed-input formulations, is studied
in \cite{HCG25}; our operational proof below uses coding and hashing rather
than invoking its resolvability theorems. These problems share the need to choose
useful auxiliary distributions and stochastic encoders. Here we use
channel symmetry to make that choice more explicit for wiretap channels.

For $n$ channel uses, a single auxiliary variable $V$ may select a
distribution over complete input blocks $X^n$. Thus $V$ labels one complete
$n$-use randomized signal. Conditioned on the physical block $x^n$, the
channel outputs are product states; the additional freedom comes from
classical correlations in the encoder. Optimizing this freedom directly is
difficult even at small blocklengths. We use covariance to describe a
family of such encoders by two distributions on a finite group.

In classical cyclic-shift symmetric wiretap channels, translated
auxiliary distributions and channel prefixing also play a structural role
\cite{OzelUlukus13}. Here the outputs are quantum states in finite-group
unitary orbits, and the construction extends to correlated preprocessing on
$G^n$ with an explicit leakage analysis.

For a transitive covariant channel, the input alphabet can be labeled,
redundantly if necessary, by the group. Choose independent variables
$V\sim P_3$ and $Z\sim P_1$, and send $X=VZ$. Here $P_1$ supplies local
noise and $P_3$ is the auxiliary distribution. Conditioned on $V=v$, the
input is a translate of $P_1$, so covariance makes each receiver's
conditional entropy independent of $v$. The marginal input distribution is
the convolution $P_3*P_1$, and the private information is
\begin{equation}
\DeltaBE(P_3*P_1)-\DeltaBE(P_1),\qquad
\DeltaBE(P):=H\Bigl(\sum_xP(x)W_x^B\Bigr)
             -H\Bigl(\sum_xP(x)W_x^E\Bigr).
\Label{eq:intro_formula}
\end{equation}
The $n$-fold channel is covariant under $G^n$, so the same construction
allows both distributions to be correlated across a finite block. We call
this the covariant preprocessing family.

The group-orbit description was used by Hayashi to analyze quantum wiretap
coding and leakage exponents \cite{MH}, and by Wu, Long and Hayashi in
secure direct communication with balanced universal hashing \cite{WLH}.
Building on this description, we retain both distributions as optimization
variables and use the convolution formula to investigate correlated
finite-block preprocessing. Uniform outer randomization gives the simpler
bound $\DeltaBE(\Pu)-\min_P\DeltaBE(P)$, because $\Pu*P=\Pu$. We show that
this bound equals the single-use private information whenever $\Pu$
maximizes $\DeltaBE$.

We also give an operational realization of the construction. The local
noise defines an induced cq wiretap channel, a constant-composition code
transmits a randomized index, and a balanced universal$_2$ hash extracts
the secret message \cite{MH,H13,WLH,TSSR11}. Resolvability-based secrecy
constructions likewise use encoder randomness to control the
eavesdropper's output distribution \cite{HK14}. Our leakage analysis
uses sandwiched R\'enyi mutual information
\cite{Wilde_2014,M_ller_Lennert_2013,HT} and applies to arbitrary $P_3$. For uniform $P_3$,
the group representation reduces the comparison-state optimization to
invariant states; in a cyclic example the minimization can be evaluated
explicitly.

For the binary channel of \cite{TBR}, Bob receives the output of a binary
symmetric channel and Eve receives one of two pure states. We compare the
covariant construction with noisy repetition, including optimization of
the logical preprocessing. At two uses, the best pairs found agree with
repetition supplemented by a symmetric logical flip. At three and four
uses, optimizing the noise distribution beyond this family gives larger
rates at every tested point where our numerical estimate of the benchmark is positive. At three
uses, the numerically obtained noise distribution has a simple interpretation: the
physical flip probability depends on whether the logical bit was flipped.
The four-use example requires a broader noise distribution. The numerical
section includes explicit feasible pairs and distinguishes the benchmark
optimizations.

Finally, we give a rigorous two-use nonadditivity example at
$p=621/5000$. An analytic convexity argument proves that the single-use
private information vanishes, while exact rational entropy bounds certify
a positive two-use rate. This provides an explicit instance of the
phenomenon established in \cite{TBR}. We also obtain the hybrid binary
channel from binary phase-shift keying when Bob uses a fixed measurement.

Sections~\ref{sec:preliminaries}--\ref{sec:covariant_wiretap} develop the
group description and preprocessing bound. Section~\ref{sec:protocol}
proves achievability and leakage bounds. The examples and numerical
comparisons are in Sections~\ref{sec:examples}--\ref{sec:TBR}, followed by
the nonadditivity witness and optical realization in
Sections~\ref{sec:witness}--\ref{sec:bpsk}.

\section{The group picture of a covariant cq channel}\Label{sec:preliminaries}

\subsection{Notation}
All Hilbert spaces are finite-dimensional except in the bosonic examples of
Section~\ref{sec:bpsk}.  We write $\calD(\calH)$ for the density operators
on $\calH$ and $H(\rho)=-\Tr\rho\log\rho$ for the von Neumann entropy; all
logarithms are natural, so information quantities are in nats, and
$h(x)=-x\log x-(1-x)\log(1-x)$ is the binary entropy.  For a classical
register $X$ with alphabet $\calX$ and a map $x\mapsto\rho^x_B$, the cq state
is $\rho_{XB}=\sum_x P_X(x)\ket{x}\!\bra{x}\otimes\rho^x_B$, and $I(X;B)$
always refers to such a state.  For a finite group $G$, $\calP(G)$ denotes
the probability distributions on $G$ and $\Pu$ the uniform distribution.

\subsection{Symmetries of a classical input alphabet}\Label{subsec:cq_covariance}

A cq channel maps a classical symbol $x$ to a density operator $W_x$.
Its quantum extension
$\calN(\rho)=\sum_x\bra{x}\rho\ket{x}W_x$ is defined relative to a
preferred basis. A symmetry of the classical input model must preserve
the diagonal algebra. It therefore maps each rank-one basis projector to
another such projector and has the form
\begin{equation}
U\ket{x}=e^{i\theta(x)}\ket{\pi(x)},
\Label{eq:monomial}
\end{equation}
where $\pi$ is a permutation. The phases have no effect on classical
states. This requirement belongs to the input model: covariance of a
quantum extension alone does not impose it, as a constant channel can be
covariant under arbitrary input unitaries.

\begin{definition}[Covariant cq channel]\Label{def:covariant}
Let a finite group $G$ act on the alphabet $\calX$ by permutations,
$(g,x)\mapsto g\cdot x$, and let $g\mapsto V_g$ be a unitary, possibly
projective, representation of $G$ on the output space.  The cq channel
$x\mapsto W_x$ is \emph{covariant} if
\begin{equation}
W_{g\cdot x}=V_g\,W_x\,V_g^\dagger\qquad\text{for all }g\in G,\ x\in\calX.
\Label{eq:cq_covariance}
\end{equation}
For a wiretap channel we require this simultaneously for Bob's outputs,
with a representation $U_B$, and for Eve's, with a representation $U_E$,
and the same permutation action on the common input alphabet.
\end{definition}

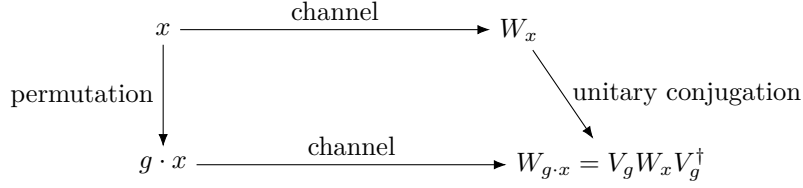
\begin{figure}[t]
\centering
\begin{tikzpicture}[node distance=1.35cm and 4.1cm,arr/.style={-Latex}]
\node (x) {$x$};
\node[right=of x] (w) {$W_x$};
\node[below=of x] (gx) {$g\cdot x$};
\node[right=of gx] (gw) {$W_{g\cdot x}=V_gW_xV_g^\dagger$};
\draw[arr] (x)--node[above]{channel}(w);
\draw[arr] (gx)--node[above]{channel}(gw);
\draw[arr] (x)--node[left]{permutation}(gx);
\draw[arr] (w)--node[right]{unitary conjugation}(gw);
\end{tikzpicture}
\caption{Covariance: permuting the classical input labels and conjugating
the output states give the same channel outputs.}
\Label{fig:cq_covariance}
\end{figure}

\subsection{Transitive symmetries, seed states and the regular action}\Label{subsec:orbit}

If the action of $G$ on $\calX$ is transitive, choose a reference symbol
$x_0$ and its output $\rho=W_{x_0}$ as a seed. Every symbol is $g\cdot x_0$
for some $g$, so covariance determines the channel from the seed.
Different group elements label the same symbol precisely when they lie
in the same coset of its stabilizer. The following statement makes the
resulting group description precise.

\begin{proposition}[Orbit structure]\Label{prop:orbit}
Let a finite group $G$ act transitively on a finite alphabet $\calX$.  For
$R\in\{B,E\}$ let $U_R$ be a possibly projective unitary representation of
$G$ with $W^R_{g\cdot x}=U_{R,g}W^R_xU_{R,g}^\dagger$.  Fix $x_0\in\calX$,
let $H=\Stab(x_0)$ and set $\rho_R:=W^R_{x_0}$.  Then $\calX\simeq G/H$ and
\begin{equation}
W^R_{g\cdot x_0}=U_{R,g}\,\rho_R\,U_{R,g}^\dagger,\qquad
U_{R,h}\,\rho_R\,U_{R,h}^\dagger=\rho_R\quad(h\in H).
\Label{eq:orbit_normal_form}
\end{equation}
Conversely, any pair of seed states invariant under $U_R|_H$ defines through
\eqref{eq:orbit_normal_form} a covariant wiretap channel on $G/H$.  Moreover,
labeling the outputs redundantly by $g\in G$ instead of by $gH\in G/H$
changes none of the finite-use private-information quantities defined in
Section~\ref{sec:wiretap_setting}.
\end{proposition}

\begin{proof}
Covariance at the reference symbol gives the first identity in
\eqref{eq:orbit_normal_form}; covariance under the stabilizer gives the
second, since $h\cdot x_0=x_0$.  For the converse, the right-hand side of
the first identity must not depend on the representative of the coset
$gH$; if $g'=gh$ with $h\in H$ then
$U_{R,g'}\rho_RU_{R,g'}^\dagger
=U_{R,g}(U_{R,h}\rho_RU_{R,h}^\dagger)U_{R,g}^\dagger
=U_{R,g}\rho_RU_{R,g}^\dagger$, where the projective phase of
$U_{R,gh}=\omega(g,h)U_{R,g}U_{R,h}$ cancels in the conjugation.  Covariance
of the channel so defined is immediate.  For the last claim, a joint
distribution of an auxiliary variable and an input block in $(G/H)^n$ lifts
to one on $G^n$ by choosing one representative for each symbol, and a
distribution on $G^n$ pushes forward to $(G/H)^n$ through
$g\mapsto g\cdot x_0$; both operations leave the induced cq states on
$VB^n$ and $VE^n$ unchanged, because those states depend on the input only
through the output states $W^R_x$.
\end{proof}

Two remarks put the proposition in context.  First, transitivity is
essential.  Under a nontransitive action the channel has one seed per orbit,
and the orbits need not be related at all: the three outputs
$\ket0\!\bra0$, $\ket1\!\bra1$ and $I/2$ are covariant under the $\bbZ_2$
action that swaps the first two labels and fixes the third, with the Pauli
$X$ acting on the output, yet they cannot all be unitary conjugates of one
seed because their spectra differ.  Second, the redundancy in the last
sentence of the proposition is harmless and useful.  By labeling the
alphabet redundantly with $G$ itself, one label for each group element, we
may assume throughout that $\calX=G$, that $G$ acts on $\calX$ by left
translation, and that
\begin{equation}
W^B_g=U_{B,g}\,\rho_B\,U_{B,g}^\dagger,\qquad
W^E_g=U_{E,g}\,\rho_E\,U_{E,g}^\dagger\qquad(g\in G).
\Label{eq:group_orbit_channel}
\end{equation}
This is the class of \emph{group-orbit channels} studied in \cite{MH,WLH},
and it is the setting of the rest of the paper.

Left translation on $G$ is transitive and free: the equation $gx=y$ has
the unique solution $g=yx^{-1}$. Such an action is called \emph{regular}.
Thus the redundant group labeling lets us work with ordinary translations
of probability distributions. Seeds inherited from a channel on $G/H$
retain their $H$-invariance; arbitrary seeds also define channels on $G$.
The results below allow projective representations, since their phases
cancel under conjugation.

\subsection{Translates, convolutions and the invariance of entropies}\Label{subsec:convolution}

For $P\in\calP(G)$ and $v\in G$, the \emph{left translate} of $P$ by $v$ is
the distribution of $vZ$ when $Z\sim P$,
\begin{equation}
(L_vP)(x):=P(v^{-1}x).
\end{equation}
For $P_3,P_1\in\calP(G)$ the \emph{convolution} is the distribution of the
product $VZ$ of independent $V\sim P_3$ and $Z\sim P_1$,
\begin{equation}
(P_3*P_1)(x):=\sum_{v\in G}P_3(v)\,P_1(v^{-1}x)=\sum_{v\in G}P_3(v)\,(L_vP_1)(x).
\Label{eq:convolution}
\end{equation}
The convolution is thus the $P_3$-mixture of the translates of $P_1$.  The
uniform distribution is absorbing, $\Pu*P=P*\Pu=\Pu$, and a point mass at
the identity is the unit.  For a group-orbit channel
\eqref{eq:group_orbit_channel} and $P\in\calP(G)$ we write
\begin{align}
\rhoBbar(P)&:=\sum_{g\in G}P(g)\,U_{B,g}\rho_BU_{B,g}^\dagger,&
\rhoEbar(P)&:=\sum_{g\in G}P(g)\,U_{E,g}\rho_EU_{E,g}^\dagger,
\Label{eq:averaged_outputs}\\
\DeltaBE(P)&:=H(\rhoBbar(P))-H(\rhoEbar(P)).&&
\Label{eq:Delta_def}
\end{align}
The function $\DeltaBE$ is the entropy difference between Bob's and Eve's
averaged output states; it is the only channel-dependent object entering the private-information
bounds.

\begin{lemma}[Translation invariance]\Label{lem:translation}
For every $P\in\calP(G)$ and $v\in G$,
\[
\rhoBbar(L_vP)=U_{B,v}\,\rhoBbar(P)\,U_{B,v}^\dagger,\qquad
\rhoEbar(L_vP)=U_{E,v}\,\rhoEbar(P)\,U_{E,v}^\dagger,
\]
and consequently $\DeltaBE(L_vP)=\DeltaBE(P)$.
\end{lemma}

\begin{proof}
Substituting $x=vz$ in the definition and using $U_{B,vz}=U_{B,v}U_{B,z}$
(up to a phase for a projective representation, which cancels in the
conjugation),
\[
\rhoBbar(L_vP)=\sum_{z\in G}P(z)\,U_{B,vz}\rho_BU_{B,vz}^\dagger
=U_{B,v}\Bigl(\sum_{z\in G}P(z)\,U_{B,z}\rho_BU_{B,z}^\dagger\Bigr)U_{B,v}^\dagger
=U_{B,v}\,\rhoBbar(P)\,U_{B,v}^\dagger,
\]
and likewise for Eve.  Unitary conjugation does not change the spectrum,
hence not the entropy.
\end{proof}

For the preprocessing family $\{L_vP_1\}_{v\in G}$, the conditional
entropies are therefore constant, whereas the mixture $P_3*P_1$ can have
different entropies. Their difference gives the private information.

\section{cq wiretap channels and the finite-use private information}\Label{sec:wiretap_setting}

A cq wiretap channel is a pair $W_B=\{W^B_x\}_{x\in\calX}$,
$W_E=\{W^E_x\}_{x\in\calX}$ of cq channels with a common classical input,
received by Bob on $\calH_B$ and by Eve on $\calH_E$
(Fig.~\ref{fig:wiretap_channel}).  Given an auxiliary classical system $V$
and a joint distribution $P_{VX}$, the relevant states are
\begin{align}
\rho_{VXB}(P_{VX})
&:=\sum_{v,x}P_{VX}(v,x)\,\ket{v}\!\bra{v}\otimes\ket{x}\!\bra{x}\otimes W^B_x,\\
\rho_{VXE}(P_{VX})
&:=\sum_{v,x}P_{VX}(v,x)\,\ket{v}\!\bra{v}\otimes\ket{x}\!\bra{x}\otimes W^E_x,
\end{align}
and the two mutual informations $I(V;B)$ and $I(V;E)$ are evaluated on them.
For $n\ge1$ let $W^{(n)}_B=\{W^B_{x_1}\otimes\cdots\otimes W^B_{x_n}\}_{x^n}$
and similarly $W^{(n)}_E$.  The $n$-use private information is
\begin{equation}
C^{(n)}(W_B,W_E):=\sup_{P_{VX^n}}\bigl[I(V;B^n)-I(V;E^n)\bigr],
\Label{eq:Cn_def}
\end{equation}
with the mutual informations evaluated on the states induced by
$P_{VX^n}$ and the $n$-fold channel.  The private capacity is
$\lim_n C^{(n)}/n$ \cite{DevetakPrivate}.  Taking products of single-use
constructions shows that
\begin{equation}
n\,C^{(1)}(W_B,W_E)\le C^{(n)}(W_B,W_E),
\Label{eq:superadditive}
\end{equation}
and the channel is said to exhibit nonadditivity at blocklength $n$ when the
inequality is strict. Conditioned on the classical block $x^n$, the
outputs are product states; the additional freedom lies in the auxiliary
variable and correlations in the classical preprocessing across the $n$
uses. A strict separation is established by an achievable $n$-use rate
per use that exceeds a converse bound on $C^{(1)}$.

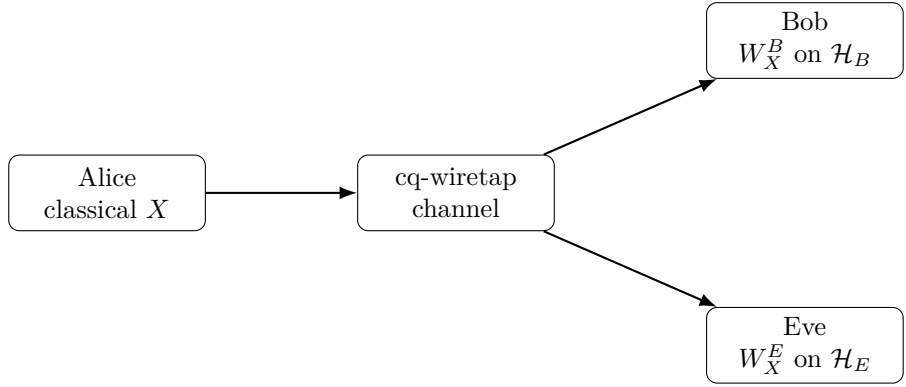
\begin{figure}[t]
\centering
\begin{tikzpicture}[
    node distance=2.0cm,
    block/.style={draw, rounded corners, align=center, minimum width=2.6cm, minimum height=1cm},
    arr/.style={-Latex, thick}
]
\node[block] (alice) {Alice\\classical $X$};
\node[block, right=of alice] (channel) {cq-wiretap\\channel};
\node[block, above right=1.0cm and 2.0cm of channel] (bob) {Bob\\$W_X^B$ on $\calH_B$};
\node[block, below right=1.0cm and 2.0cm of channel] (eve) {Eve\\$W_X^E$ on $\calH_E$};
\draw[arr] (alice) -- (channel);
\draw[arr] (channel) -- (bob);
\draw[arr] (channel) -- (eve);
\end{tikzpicture}
\caption{A cq wiretap channel has one classical input and two quantum outputs.  The coding problem compares the information available to Bob with the information leaked to Eve.}
\Label{fig:wiretap_channel}
\end{figure}

For a binary input alphabet the single-use quantity has a transparent
description that we will use twice.  Write $\DeltaBE(x)$ for the entropy
difference \eqref{eq:Delta_def} of the input distribution $\Bern(x)$ that
puts probability $x$ on the symbol $1$, and let $\DeltaBE^{\mathrm{vex}}$ be
its lower convex envelope on $[0,1]$, the largest convex function below
$\DeltaBE$.

\begin{lemma}[Single-use private information of a binary-input channel]\Label{lem:binary_single_use}
For any cq wiretap channel with $\calX=\{0,1\}$,
\begin{equation}
C^{(1)}(W_B,W_E)=\max_{x\in[0,1]}\bigl[\DeltaBE(x)-\DeltaBE^{\mathrm{vex}}(x)\bigr].
\Label{eq:binary_single_use}
\end{equation}
In particular $C^{(1)}=0$ if and only if $\DeltaBE$ is convex on $[0,1]$.
\end{lemma}

\begin{proof}
For a joint distribution $P_{VX}$ with conditionals
$P_{X|V=v}=\Bern(x_v)$ and marginal $P_X=\Bern(\bar x)$,
$\bar x=\sum_vP_V(v)x_v$, a direct computation gives
$I(V;B)-I(V;E)=\DeltaBE(\bar x)-\sum_vP_V(v)\DeltaBE(x_v)$.  For fixed
$\bar x$, the infimum of $\sum_vP_V(v)\DeltaBE(x_v)$ over all finitely
supported probability measures on $[0,1]$ with mean $\bar x$ is the value at
$\bar x$ of the lower convex envelope, and it is attained by a two-point
measure because $\DeltaBE$ is continuous on a compact interval.  Maximizing
over $\bar x$ gives \eqref{eq:binary_single_use}.
\end{proof}

\section{The covariant preprocessing bound}\Label{sec:covariant_wiretap}

We now assume that $(W_B,W_E)$ is a group-orbit channel
\eqref{eq:group_orbit_channel} with input alphabet $G$, and we use the
notation of Section~\ref{subsec:convolution}.  The construction is the one
sketched in the introduction and drawn in Fig.~\ref{fig:convolution_bound}:
an outer variable $V\sim P_3$ and an independent noise variable
$Z\sim P_1$ are multiplied in the group, and $X=VZ$ is sent.

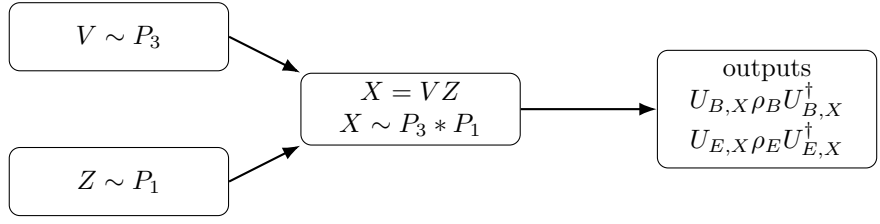
\begin{figure}[t]
\centering
\begin{tikzpicture}[
    node distance=1.5cm,
    block/.style={draw, rounded corners, align=center, minimum width=2.9cm, minimum height=.9cm},
    arr/.style={-Latex, thick}
]
\node[block] (v) {$V\sim P_3$};
\node[block, below=1.0cm of v] (y) {$Z\sim P_1$};
\node[block, right=2.4cm of $(v)!0.5!(y)$] (x) {$X=VZ$\\$X\sim P_3*P_1$};
\node[block, right=1.8cm of x] (out) {outputs\\$U_{B,X}\rho_B U_{B,X}^\dagger$\\$U_{E,X}\rho_E U_{E,X}^\dagger$};
\draw[arr] (v.east) -- (x.north west);
\draw[arr] (y.east) -- (x.south west);
\draw[arr] (x) -- (out);
\end{tikzpicture}
\caption{The covariant preprocessing.  Conditioned on $V=v$ the input is the translate $L_vP_1$, whose entropies do not depend on $v$; the input marginal is the convolution $P_3*P_1$.}
\Label{fig:convolution_bound}
\end{figure}

\begin{theorem}[Covariant preprocessing bound]\Label{thm:main}
Let $(W_B,W_E)$ be a group-orbit channel with input alphabet $G$, and let
$P_1,P_3\in\calP(G)$.  For independent $V\sim P_3$ and $Z\sim P_1$ and the
input $X=VZ$, that is, for the joint distribution
\begin{equation}
P_{VX}(v,x)=P_3(v)\,P_1(v^{-1}x),
\Label{eq:PVX}
\end{equation}
one has the exact identity
\begin{equation}
I(V;B)-I(V;E)=\DeltaBE(P_3*P_1)-\DeltaBE(P_1).
\Label{eq:main_identity}
\end{equation}
Consequently
\begin{equation}
C^{(1)}(W_B,W_E)\ \ge\ \max_{P_1,P_3\in\calP(G)}\bigl[\DeltaBE(P_3*P_1)-\DeltaBE(P_1)\bigr].
\Label{eq:main_bound}
\end{equation}
The choice $P_3=\Pu$ gives
\begin{equation}
C^{(1)}(W_B,W_E)\ \ge\ \DeltaBE(\Pu)-\min_{P\in\calP(G)}\DeltaBE(P),
\Label{eq:uniform_bound}
\end{equation}
and \eqref{eq:uniform_bound} holds with equality whenever
$\Pu\in\argmax_{P\in\calP(G)}\DeltaBE(P)$.
\end{theorem}

\begin{proof}
Under \eqref{eq:PVX} the marginal of $X$ is $P_3*P_1$ by
\eqref{eq:convolution}, and the conditional distribution of $X$ given
$V=v$ is the translate $L_vP_1$.  Bob's conditional state given $V=v$ is
therefore $\rhoBbar(L_vP_1)$, whose entropy equals $H(\rhoBbar(P_1))$ by
Lemma~\ref{lem:translation}, and his average state is $\rhoBbar(P_3*P_1)$.
Hence
$I(V;B)=H(\rhoBbar(P_3*P_1))-\sum_vP_3(v)H(\rhoBbar(L_vP_1))
=H(\rhoBbar(P_3*P_1))-H(\rhoBbar(P_1))$, and the same computation for Eve
gives $I(V;E)=H(\rhoEbar(P_3*P_1))-H(\rhoEbar(P_1))$.  Subtracting yields
\eqref{eq:main_identity}, and \eqref{eq:main_bound} follows from the
definition of $C^{(1)}$ as a supremum.  Since $\Pu*P_1=\Pu$ for every
$P_1$, the choice $P_3=\Pu$ followed by a minimization over $P_1$ gives
\eqref{eq:uniform_bound}.

For the equality statement, let $P_{VX}$ be arbitrary with conditionals
$P_{X|V=v}$ and marginal $P_X$.  Then
\[
I(V;B)-I(V;E)=\DeltaBE(P_X)-\sum_vP_V(v)\,\DeltaBE(P_{X|V=v})
\le\max_P\DeltaBE(P)-\min_P\DeltaBE(P),
\]
and if $\Pu$ maximizes $\DeltaBE$, this bound coincides with the right-hand
side of \eqref{eq:uniform_bound}.  Combined with \eqref{eq:uniform_bound}
this gives equality.
\end{proof}

On the binary parameter line studied in Section~\ref{sec:TBR}, the
uniform input does not maximize $\DeltaBE$. There we use the binary
convex-envelope formula to evaluate the single-use private information.

Because the $n$-fold channel is again a group-orbit channel, for the product
group $G^n$ with product seed states and product representations, the
theorem applies at every blocklength.  Write $\DeltaBE_n$ for the entropy
difference functional \eqref{eq:Delta_def} of the $n$-fold channel, defined
on $\calP(G^n)$.

\begin{corollary}[Finite-block covariant preprocessing]\Label{cor:finite_block_covariant}
For every $n\ge1$ and all $P_1,P_3\in\calP(G^n)$,
\begin{equation}
\frac1n\,C^{(n)}(W_B,W_E)\ \ge\ \frac1n\bigl[\DeltaBE_n(P_3*P_1)-\DeltaBE_n(P_1)\bigr],
\Label{eq:finite_block_covariant}
\end{equation}
so that every pair $(P_1,P_3)$ gives an explicit achievable bound at
blocklength $n$.  We write
\begin{equation}
R_n^{\pair}:=\frac1n\max_{P_1,P_3\in\calP(G^n)}\bigl[\DeltaBE_n(P_3*P_1)-\DeltaBE_n(P_1)\bigr]
\Label{eq:Rpair}
\end{equation}
for the best value within the family.
\end{corollary}

The conditional input distributions in this family are translates of a
common noise distribution; the bound does not require this restriction
to be optimal among all auxiliary-variable constructions. Both $P_1$ and
$P_3$ can contain correlations across uses. Uniform $P_3$ gives the
product-channel version of \eqref{eq:uniform_bound}. For $G^n=\bbZ_2^n$,
uniform $P_3$ on $\{0^n,1^n\}$ and product noise give noisy repetition.
A fixed preprocessing block is used as a super-symbol in a longer channel
code, as described next.

\section{Operational realization of the covariant bound}\Label{sec:protocol}

Treat a fixed preprocessing block as one super-symbol and use the
induced channel $N$ times. Throughout this section, $G$ denotes the
super-symbol group, which is $G^n$ when the preprocessing block has
length $n$, and $N$ is the coding blocklength.

\subsection{Induced channel and information difference}
Fix $P_1,P_3\in\calP(G)$ and define the sender-randomized channels
\begin{equation}
\widetilde W^B_v:=\sum_{z\in G}P_1(z)\,W^B_{vz}=\rhoBbar(L_vP_1),\qquad
\widetilde W^E_v:=\sum_{z\in G}P_1(z)\,W^E_{vz}=\rhoEbar(L_vP_1).
\Label{eq:induced_channel}
\end{equation}
When $V\sim P_3$ is fed into these channels, the physical input has
distribution $P_3*P_1$, and Theorem~\ref{thm:main} says exactly that
\begin{equation}
I(V;B)_{\widetilde W^B}-I(V;E)_{\widetilde W^E}=\DeltaBE(P_3*P_1)-\DeltaBE(P_1).
\Label{eq:induced_difference}
\end{equation}
The task of the coding layer is to realize this difference as a rate of
secret communication over the induced channel.

\subsection{Codes, hashes and conventions}
The secret message $M$ is uniformly distributed on
$\{1,\dots,\sM_N\}$.  A \emph{balanced universal$_2$ family} is a finite
collection $\{F_s\}_{s\in\calS}$ of surjective maps
$F_s:\{1,\dots,\sK_N\}\to\{1,\dots,\sM_N\}$, $\sK_N=\sM_N\sL_N$, such that
every fiber $F_s^{-1}(m)$ has exactly $\sL_N$ elements and, for any two
distinct indices $m'\neq m''$, $\Pr_s[F_s(m')=F_s(m'')]\le 1/\sM_N$ when the
seed $s$ is drawn uniformly from $\calS$ \cite{CW,H13,WLH}.  Universality is a
property of the family, not of a single map; the seed is public and is
drawn independently of everything else.  We write $\bbE_F$ for the average
over the seed and
\begin{equation}
\rho_{\mathrm{sac},N}:=\frac1N\log\sL_N
\Label{eq:sacrifice_rate}
\end{equation}
for the sacrifice rate.

The channel code $\phi_N:\{1,\dots,\sK_N\}\to G^N$,
$\phi_N(m')=(\phi_{N,1}(m'),\dots,\phi_{N,N}(m'))$, is a code for Bob's
induced channel $\widetilde W^B$.  We take it to be a
\emph{constant-composition} code: all codewords have the same empirical
type $Q_N\in\calP(G)$, with $\supp Q_N\subseteq S:=\supp P_3$ and
$Q_N\to P_3$ as $N\to\infty$.  Constant-composition codes achieve every rate
below the mutual information of a finite-input memoryless cq channel; see
\cite{Holevo98,SW97,Winter99,Hayashi_universal} and \cite[Chapter~4]{Springer}.  Decoding
success is measured by the average error probability $\epsilon^B_N$ over a
uniformly distributed index $M'$.  Secrecy is measured by the trace
distance
\begin{equation}
\delta^E_N:=\bigl\|\rho_{ME^N}-\rho_M\otimes\rho_{E^N}\bigr\|_1,
\Label{eq:leakage_criterion}
\end{equation}
where $\rho_{ME^N}$ is the joint state of the uniform message and Eve's
$N$ output systems.

\subsection{The protocol}
\begin{enumerate}
\item To send $m$, Alice draws $M'$ uniformly from the fiber
$F_s^{-1}(m)$.
\item She encodes it as $\phi_N(M')\in G^N$.
\item Independently for each position $j$ she draws $Z_j\sim P_1$ and sends
$X_j=\phi_{N,j}(M')\,Z_j$ through the physical channel.
\item Bob decodes $\widehat M'$ with the decoder of the code $\phi_N$ for
$\widetilde W^B$ and outputs $F_s(\widehat M')$.
\end{enumerate}
The code and the hash seed are public.  Only the local randomization
$Z^N$ and the draw inside the hash fiber are private
(Fig.~\ref{fig:protocol}).  Because the fibers are balanced, $M'$ is
uniformly distributed on $\{1,\dots,\sK_N\}$ for every seed $s$; this is
what allows us later to select a seed on the basis of its leakage without
disturbing the error analysis of the code.

\begin{figure}[t]
\centering
\resizebox{0.96\linewidth}{!}{%
\begin{tikzpicture}[
    node distance=1.15cm,
    block/.style={draw, rounded corners, align=center, minimum width=2.45cm, minimum height=.85cm},
    arr/.style={-Latex, thick}
]
\node[block] (msg) {message\\$M$};
\node[block, right=of msg] (fiber) {uniform draw\\$M'\in F_s^{-1}(M)$};
\node[block, right=of fiber] (code) {channel code\\$\phi_N(M')\in G^N$};
\node[block, right=of code] (noise) {local randomization\\$Z_j\sim P_1$};
\node[block, right=of noise] (input) {channel input\\$X_j=\phi_{N,j}(M')Z_j$};
\draw[arr] (msg) -- (fiber);
\draw[arr] (fiber) -- (code);
\draw[arr] (code) -- (noise);
\draw[arr] (noise) -- (input);
\end{tikzpicture}%
}
\caption{Encoding for the covariant bound.  The hash seed determines the fiber from which the randomized index is drawn, and $P_1$ supplies the sender-side group randomization.}
\Label{fig:protocol}
\end{figure}
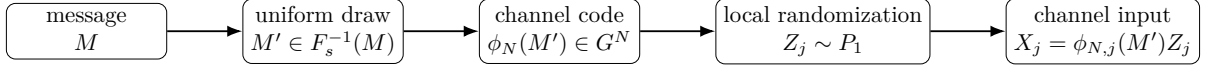

\subsection{Reliability}
The constant-composition coding theorem for the memoryless channel
$\widetilde W^B$ allows us to choose the sequence $\{\phi_N\}$, with types
$Q_N\to P_3$ supported in $\supp P_3$, so that $\epsilon^B_N\to0$ whenever
the total code rate satisfies
\begin{equation}
\limsup_{N\to\infty}\frac1N\log\sK_N<I(V;B)_{\widetilde W^B}.
\Label{eq:reliability_rate}
\end{equation}

\subsection{Secrecy for an arbitrary outer distribution}\Label{subsec:secrecy_general}
For $\alpha>1$, the sandwiched R\'enyi divergence is
\cite{Wilde_2014,M_ller_Lennert_2013}
\begin{equation}
D_\alpha(\rho\|\sigma):=\frac1{\alpha-1}\log\Tr\Bigl\{\bigl(\sigma^{\frac{1-\alpha}{2\alpha}}\rho\,\sigma^{\frac{1-\alpha}{2\alpha}}\bigr)^{\alpha}\Bigr\},
\end{equation}
where negative powers of $\sigma$ are taken on its support, and the
divergence is $+\infty$ unless $\supp\rho\subseteq\supp\sigma$.  It is
monotone in $\alpha$ and converges to the relative entropy $D(\rho\|\sigma)$
as $\alpha\downarrow1$.  The associated R\'enyi mutual information of a state
$\rho_{AE}$ is \cite{HT}
\begin{equation}
I_\alpha(A;E)[\rho_{AE}]:=\min_{\sigma_E}D_\alpha(\rho_{AE}\|\rho_A\otimes\sigma_E),
\end{equation}
and for a cq state $\sum_aP(a)\ket a\!\bra a\otimes\rho_{E|a}$ it reduces to
\begin{equation}
I_\alpha(A;E)=\min_{\sigma_E}\frac1{\alpha-1}\log\sum_aP(a)\Tr\Bigl\{\bigl(\sigma_E^{\frac{1-\alpha}{2\alpha}}\rho_{E|a}\sigma_E^{\frac{1-\alpha}{2\alpha}}\bigr)^{\alpha}\Bigr\}.
\Label{eq:renyi_cq}
\end{equation}
Let $\rho_{M'E^N}$ be the state of the uniform index $M'$ and Eve's outputs,
\begin{equation}
\rho_{M'E^N}:=\sum_{m'=1}^{\sK_N}\frac1{\sK_N}\ket{m'}\!\bra{m'}\otimes
\bigotimes_{j=1}^N\widetilde W^E_{\phi_{N,j}(m')},
\end{equation}
and write $I_\alpha(M';E^N)$ for its R\'enyi mutual information.  The
balanced-hashing bound of \cite[Theorem~1]{WLH}, expressed in nats, gives
\begin{equation}
\bbE_F\,\delta^E_N\ \le\ \min_{0\le t\le1}2^{\frac{1-t}{1+t}}
\exp\Bigl[-\frac{t}{1+t}\bigl(N\rho_{\mathrm{sac},N}-I_{1+t}(M';E^N)\bigr)\Bigr].
\Label{EE1}
\end{equation}
A fixed positive secrecy exponent follows from the common codeword
type. For $S=\supp P_3$, define
\begin{equation}
\omega_v:=\widetilde W^E_v,\qquad
\sigma_E:=\sum_{v\in S}P_3(v)\,\omega_v,\qquad
J_\alpha(Q):=\sum_{v\in S}Q(v)\,D_\alpha(\omega_v\|\sigma_E)\quad(Q\in\calP(S)).
\Label{eq:J_def}
\end{equation}
Every $\omega_v$ with $v\in S$ is supported in $\supp\sigma_E$, so all the
divergences in \eqref{eq:J_def} are finite.

\begin{lemma}[Type bound on the R\'enyi leakage]\Label{lem:type_bound}
If every codeword of $\phi_N$ has type $Q_N$ with $\supp Q_N\subseteq S$,
then for $1<\alpha\le2$
\begin{equation}
I_\alpha(M';E^N)\le N\,J_\alpha(Q_N).
\Label{eq:type_renyi}
\end{equation}
\end{lemma}

\begin{proof}
Use the product state $\sigma_E^{\otimes N}$ as the comparison state in
\eqref{eq:renyi_cq}.  Since the conditional output of codeword $m'$ is the
product $\bigotimes_j\omega_{\phi_{N,j}(m')}$ and the comparison state is a
product, the trace factorizes over the $N$ positions:
\begin{align}
I_\alpha(M';E^N)
&\le\frac1{\alpha-1}\log\Biggl\{\frac1{\sK_N}\sum_{m'=1}^{\sK_N}
\prod_{j=1}^N\Tr\Bigl\{\bigl(\sigma_E^{\frac{1-\alpha}{2\alpha}}\omega_{\phi_{N,j}(m')}\sigma_E^{\frac{1-\alpha}{2\alpha}}\bigr)^{\alpha}\Bigr\}\Biggr\}\notag\\
&=\frac1{\alpha-1}\log\Biggl\{\frac1{\sK_N}\sum_{m'=1}^{\sK_N}
\exp\Bigl[(\alpha-1)\sum_{j=1}^N D_\alpha\bigl(\omega_{\phi_{N,j}(m')}\big\|\sigma_E\bigr)\Bigr]\Biggr\}.
\end{align}
For a codeword of type $Q_N$ the inner sum equals
$N\sum_vQ_N(v)D_\alpha(\omega_v\|\sigma_E)=NJ_\alpha(Q_N)$, the same for
every $m'$, and the average over $m'$ is trivial.
\end{proof}

\begin{lemma}[Continuity at $\alpha=1$]\Label{lem:J_continuity}
$\lim_{t\downarrow0}J_{1+t}(P_3)=\sum_{v\in S}P_3(v)D(\omega_v\|\sigma_E)=I(V;E)_{\widetilde W^E}$.
Moreover, for each fixed $t\in(0,1]$ the map $Q\mapsto J_{1+t}(Q)$ is
continuous on $\calP(S)$.
\end{lemma}

\begin{proof}
In finite dimension and under the support condition,
$D_{1+t}(\omega_v\|\sigma_E)\to D(\omega_v\|\sigma_E)$ as $t\downarrow0$
\cite{M_ller_Lennert_2013}; the sum over the finite set $S$ then converges,
and $\sum_vP_3(v)D(\omega_v\|\sigma_E)$ is the mutual information of the
cq state with conditionals $\omega_v$ and average $\sigma_E$.  The second
statement holds because $J_{1+t}$ is a finite linear combination of
finite numbers.
\end{proof}

\begin{theorem}[Achievability of the covariant bound]\Label{thm:achievability}
Let $(P_1,P_3)$ be any pair in $\calP(G)\times\calP(G)$ and let
$\{\phi_N\}$ be constant-composition codes for $\widetilde W^B$ with types
$Q_N\to P_3$, $\supp Q_N\subseteq\supp P_3$.  If
\begin{equation}
\liminf_{N\to\infty}\rho_{\mathrm{sac},N}>I(V;E)_{\widetilde W^E},
\Label{eq:sacrifice_condition}
\end{equation}
then $\bbE_F\delta^E_N\to0$ exponentially in $N$. In particular, with
codes provided by the coding theorem, every nonnegative rate $R<\DeltaBE(P_3*P_1)-\DeltaBE(P_1)$ is achievable with
vanishing average error and trace-distance leakage. For each $N$, a
deterministic public hash seed $s_N$ can be chosen.
\end{theorem}

\begin{proof}
Choose $\eta>0$ with
$\liminf_N\rho_{\mathrm{sac},N}>I(V;E)_{\widetilde W^E}+3\eta$.  By
Lemma~\ref{lem:J_continuity} there is a fixed $t\in(0,1]$ with
$J_{1+t}(P_3)<I(V;E)_{\widetilde W^E}+\eta$, and by the continuity in $Q$
and $Q_N\to P_3$ there is $N_0$ such that for all $N\ge N_0$
\[
J_{1+t}(Q_N)<I(V;E)_{\widetilde W^E}+2\eta,
\qquad
\rho_{\mathrm{sac},N}>I(V;E)_{\widetilde W^E}+3\eta.
\]
Lemma~\ref{lem:type_bound} then gives
$N\rho_{\mathrm{sac},N}-I_{1+t}(M';E^N)\ge N\bigl(\rho_{\mathrm{sac},N}-J_{1+t}(Q_N)\bigr)>N\eta$,
and \eqref{EE1} with this fixed $t$ yields
\begin{equation}
\bbE_F\,\delta^E_N\le 2^{\frac{1-t}{1+t}}\exp\Bigl[-\frac{Nt\eta}{1+t}\Bigr]\longrightarrow0.
\end{equation}
For the operational statement, fix a nonnegative rate
$R<I(V;B)_{\widetilde W^B}-I(V;E)_{\widetilde W^E}$ and a sacrifice rate
$R_{\mathrm{sac}}$ with
$I(V;E)_{\widetilde W^E}<R_{\mathrm{sac}}<I(V;B)_{\widetilde W^B}-R$, and
take $\sM_N=\lfloor e^{NR}\rfloor$ and $\sL_N=\lceil e^{NR_{\mathrm{sac}}}\rceil$.
The total rate $R+R_{\mathrm{sac}}$ is below $I(V;B)_{\widetilde W^B}$, so
reliable constant-composition codes of this size exist by the coding
theorem, and \eqref{eq:sacrifice_condition} holds with the fixed margin
$R_{\mathrm{sac}}-I(V;E)_{\widetilde W^E}>0$, so the leakage vanishes
exponentially by the first part.  Finally, since the leakage bound is an
average over the seed, for each $N$ some seed $s_N$ has
$\delta^E_N(s_N)\le\bbE_F\delta^E_N$.  Fixing $s_N$ does not change the
distribution of the index $M'$, which is uniform for every seed because the
fibers are balanced, so the average decoding error of the code is unchanged
and reliability and secrecy hold together.  Taking the supremum over $R$
gives every rate below the difference \eqref{eq:induced_difference}.
\end{proof}

Optimizing over $(P_1,P_3)$ in Theorem~\ref{thm:achievability} recovers the
bound \eqref{eq:main_bound} operationally, and applying it to the product
group gives Corollary~\ref{cor:finite_block_covariant}.

\begin{remark}[Worst-case message]\Label{rem:worst_case}
Theorem~\ref{thm:achievability} is stated for a uniform message, which is
the criterion we use.  A statement for every message distribution follows
by expurgation: keep the messages whose conditional decoding error is at
most $\sqrt{\epsilon^B_N}$ and whose conditional Eve state is within trace
distance $\sqrt{\delta^E_N}$ of the average Eve state.  By Markov's
inequality at least a fraction $1-\sqrt{\epsilon^B_N}-\sqrt{\delta^E_N}$ of
the messages survive, the rate is unchanged asymptotically, and for every
distribution on the retained set the leakage relative to its own Eve
marginal is at most $2\sqrt{\delta^E_N}$.
For modular constructions of semantically secure classical--quantum and
quantum wiretap codes, see \cite{BCDFW22}.
\end{remark}

\subsection{Uniform outer randomization and an explicit exponent}\Label{subsec:uniform_exponent}

When the outer distribution is uniform, $P_3=\Pu$, the group structure
gives a single-letter bound on the R\'enyi leakage directly, without any
condition on the composition of the code, and the twirl over $G$ reduces the
minimization over comparison states to a minimization over a few
parameters.  Define
\begin{equation}
\rho_{E|P_1}:=\sum_{z\in G}P_1(z)\,U_{E,z}\rho_EU_{E,z}^\dagger=\rhoEbar(P_1),
\qquad
\rho_{AE}(P_1):=\sum_{x\in G}\frac1{|G|}\ket x\!\bra x\otimes U_{E,x}\,\rho_{E|P_1}\,U_{E,x}^\dagger.
\end{equation}
The state $\rho_{AE}(P_1)$ is the cq state of the induced channel
$\widetilde W^E$ with uniform input.

\begin{lemma}\Label{L2}
For every code $\phi_N$ and every $\alpha>1$,
\begin{equation}
I_\alpha(M';E^N)\le N\,I_\alpha(A;E)[\rho_{AE}(P_1)].
\end{equation}
\end{lemma}

\begin{proof}
Let $\tau_E:=\rho_{E|P_1}$ and $U^{(N)}_{E,x}:=U_{E,x_1}\otimes\cdots\otimes U_{E,x_N}$
for $x\in G^N$, so that the conditional state of Eve for the codeword
$\phi_N(m')$ is
\[
\tau_{m'}:=U^{(N)}_{E,\phi_N(m')}\,\tau_E^{\otimes N}\,(U^{(N)}_{E,\phi_N(m')})^\dagger .
\]  Let
$\sigma^*_E$ be an optimizer in the single-copy definition of $I_\alpha(A;E)[\rho_{AE}(P_1)]$.  The
product state $(\sigma^*_E)^{\otimes N}$ is an admissible comparison state
for the $N$-copy orbit, and the orbit sum factorizes.
Write $\Sigma:=(\sigma^*_E)^{\otimes N}$ and $\gamma:=\frac{1-\alpha}{2\alpha}$.
Then
\begin{equation}
e^{(\alpha-1)NI_\alpha(A;E)[\rho_{AE}(P_1)]}
=\sum_{x\in G^N}\frac1{|G|^N}\Tr\Bigl\{\Bigl(\Sigma^{\gamma}\,
U^{(N)}_{E,x}\tau_E^{\otimes N}(U^{(N)}_{E,x})^\dagger\,\Sigma^{\gamma}\Bigr)^{\alpha}\Bigr\}.
\end{equation}
Since $x\,\phi_N(m')$ is uniformly
distributed on $G^N$ when $x$ is, the right-hand side is unchanged if $x$ is
replaced by $x\,\phi_N(m')$ and the result averaged over $m'$:
\begin{align}
e^{(\alpha-1)NI_\alpha(A;E)[\rho_{AE}(P_1)]}
&=\sum_{x\in G^N}\frac1{|G|^N}\sum_{m'=1}^{\sK_N}\frac1{\sK_N}
\Tr\Bigl\{\Bigl(\Sigma^{\gamma}\,U^{(N)}_{E,x}\,\tau_{m'}\,(U^{(N)}_{E,x})^\dagger\,\Sigma^{\gamma}\Bigr)^{\alpha}\Bigr\}\notag\\
&=\sum_{x\in G^N}\frac1{|G|^N}\sum_{m'=1}^{\sK_N}\frac1{\sK_N}
\Tr\Bigl\{\Bigl(\Sigma_x^{\gamma}\,\tau_{m'}\,\Sigma_x^{\gamma}\Bigr)^{\alpha}\Bigr\},
\qquad \Sigma_x:=(U^{(N)}_{E,x})^\dagger\,\Sigma\,U^{(N)}_{E,x},
\end{align}
where $\Sigma_x$ is again a state.  For each $x$ the inner sum is at least the minimum over
all comparison states, which is $e^{(\alpha-1)I_\alpha(M';E^N)}$, and
averaging over $x$ preserves the inequality.
\end{proof}

For $P_3=\Pu$, the single-copy bound converges to
$I(V;E)_{\widetilde W^E}$ as $\alpha\downarrow1$, yielding exponential
secrecy without a composition restriction whenever
$\liminf_N\rho_{\mathrm{sac},N}>I(V;E)_{\widetilde W^E}$.
We next simplify the
comparison-state optimization using the group twirl.
Decompose $\calH_E$ into irreducible subspaces of
$U_E$,
\begin{equation}
\calH_E=\calH_{E,1}\oplus\cdots\oplus\calH_{E,\sJ},
\end{equation}
and assume for simplicity that these irreducible subspaces are pairwise
inequivalent.  Let $\rho_j$ be the completely mixed state on $\calH_{E,j}$
and, for a distribution $P_J$ on $\{1,\dots,\sJ\}$, let
$\sigma_E(P_J):=\sum_jP_J(j)\rho_j$.

\begin{lemma}\Label{L3}
Under the multiplicity-free assumption, for $1<\alpha\le2$,
\begin{equation}
I_\alpha(A;E)[\rho_{AE}(P_1)]
=\frac1{\alpha-1}\log\min_{P_J}\Tr\Bigl\{\Bigl(\rho_{E|P_1}^{\frac12}\,\sigma_E(P_J)^{\frac{1-\alpha}{\alpha}}\,\rho_{E|P_1}^{\frac12}\Bigr)^{\alpha}\Bigr\}.
\end{equation}
The value at $\alpha=1$ is the limit $\alpha\downarrow1$; negative powers
use the support convention stated above.
\end{lemma}

\begin{proof}
Let $\tau:=\rho_{E|P_1}$, $r:=\frac{1-\alpha}{\alpha}\in[-\tfrac12,0)$, and
write $\bbE_x$ for the average over uniformly distributed $x\in G$.  For
$1<\alpha\le2$, using $\Tr\{(\sigma^{r/2}A\sigma^{r/2})^\alpha\}=\Tr\{(A^{1/2}\sigma^{r}A^{1/2})^\alpha\}$
for $A\ge0$,
\begin{align}
e^{(\alpha-1)I_\alpha(A;E)[\rho_{AE}(P_1)]}
&=\min_{\sigma_E}\ \bbE_x\Tr\Bigl\{\Bigl(\tau^{\frac12}U_{E,x}^\dagger\sigma_E^{\,r}U_{E,x}\tau^{\frac12}\Bigr)^{\alpha}\Bigr\}
\notag\\
&\stackrel{(a)}{\ge}\min_{\sigma_E}\Tr\Bigl\{\Bigl(\tau^{\frac12}\bigl(\bbE_xU_{E,x}^\dagger\sigma_E^{\,r}U_{E,x}\bigr)\tau^{\frac12}\Bigr)^{\alpha}\Bigr\}
\notag\\
&\stackrel{(b)}{\ge}\min_{\sigma_E}\Tr\Bigl\{\Bigl(\tau^{\frac12}\bigl(\bbE_xU_{E,x}^\dagger\sigma_EU_{E,x}\bigr)^{r}\tau^{\frac12}\Bigr)^{\alpha}\Bigr\}
\notag\\
&\stackrel{(c)}{=}\min_{P_J}\Tr\Bigl\{\Bigl(\tau^{\frac12}\sigma_E(P_J)^{r}\tau^{\frac12}\Bigr)^{\alpha}\Bigr\}
\stackrel{(d)}{=}\min_{P_J}\ \bbE_x\Tr\Bigl\{\Bigl(\tau^{\frac12}U_{E,x}^\dagger\sigma_E(P_J)^{r}U_{E,x}\tau^{\frac12}\Bigr)^{\alpha}\Bigr\}
\notag\\
&\ge\min_{\sigma_E}\ \bbE_x\Tr\Bigl\{\Bigl(\tau^{\frac12}U_{E,x}^\dagger\sigma_E^{\,r}U_{E,x}\tau^{\frac12}\Bigr)^{\alpha}\Bigr\}.
\end{align}
Step $(a)$ is the convexity of $X\mapsto\Tr X^\alpha$ on positive operators
for $1<\alpha\le2$.  To handle comparison states that are not of full rank,
replace $\sigma_E$ by $(1-\varepsilon)\sigma_E+\varepsilon I/\dim\calH_E$,
apply the chain, and let $\varepsilon\downarrow0$ using the support
convention and lower semicontinuity.  For $(b)$, the map $t\mapsto t^{r}$ is
operator convex for $r\in[-1,0]$, so
$\bbE_xU_{E,x}^\dagger\sigma_E^{\,r}U_{E,x}\ge(\bbE_xU_{E,x}^\dagger\sigma_EU_{E,x})^{r}$,
and multiplying by $\tau^{1/2}$ on both sides and using
\cite[Lemma~A.13]{Springer} gives the trace inequality.  Because the
irreducible subspaces are pairwise inequivalent, the twirl
$\bbE_xU_{E,x}^\dagger\sigma_EU_{E,x}$ is of the form
$\sigma_E(P_J)$, which is $(c)$; and $\sigma_E(P_J)$ commutes with every
$U_{E,x}$, which is $(d)$.  The chain begins and ends with the same
quantity, so all inequalities are equalities.
\end{proof}

\begin{remark}[Multiplicities]
Without the multiplicity-free assumption the twirl of a state has the form
$\bigoplus_j(I_{d_j}/d_j)\otimes A_j$, with $A_j\ge0$ acting on the
multiplicity space of the $j$-th irreducible component and
$\sum_j\Tr A_j=1$.  Lemma~\ref{L3} then holds with the minimum taken over
such states instead of over $P_J$; we do not need this extension.  In the
binary phase-shift-keying examples of Section~\ref{sec:bpsk} the two output
states span an invariant subspace of dimension at most two, so the
finite-dimensional calculations apply even though the ambient space is a
Fock space. To restrict the comparison-state minimization as well,
let $S$ project onto the joint support of the two output states and choose
a state $\omega$ supported on $S$. The channel
$\mathcal T(\sigma)=S\sigma S+\Tr[(I-S)\sigma]\omega$ fixes each output
state. Then, data processing for the sandwiched R\'enyi divergence ($\alpha>1$)
gives
$D_\alpha(\rho_{AE}\|\rho_A\otimes\mathcal T(\sigma))
\le D_\alpha(\rho_{AE}\|\rho_A\otimes\sigma)$,
so the minimization may be performed on $S$. 
For nonzero BPSK amplitudes,
the restricted parity representation has one even and one odd component.
\end{remark}

Combining Lemmas~\ref{L2} and~\ref{L3} with \eqref{EE1}, the leakage of the
uniform scheme satisfies
\begin{equation}
\bbE_F\,\delta^E_N\ \le\ \min_{0\le t\le1}2^{\frac{1-t}{1+t}}
\exp\Bigl[-\frac{t}{1+t}N\rho_{\mathrm{sac},N}\Bigr]
\Biggl(\min_{P_J}\Tr\Bigl\{\Bigl(\rho_{E|P_1}^{\frac12}\,\sigma_E(P_J)^{\frac{-t}{1+t}}\,\rho_{E|P_1}^{\frac12}\Bigr)^{1+t}\Bigr\}\Biggr)^{\frac{N}{1+t}}.
\Label{EE4}
\end{equation}

\section{Examples}\Label{sec:examples}

\subsection{Cyclic diagonal representation}\Label{subsec:cyclic_example}

Let $G=\bbZ_d$ and let $\calH_E$ be $d$-dimensional with basis
$\{\ket j\}_{j=0}^{d-1}$ and representation $U(x)\ket j=e^{2\pi i jx/d}\ket j$.
Take a pure seed
\[
\ket{\psi_E}=\sum_{j=0}^{d-1}\sqrt{q_j}\,\ket j,\qquad \rho_E=\ket{\psi_E}\!\bra{\psi_E},
\]
with $q$ a distribution on $\{0,\dots,d-1\}$.  The irreducible subspaces of
$U$ are the lines $\bbC\ket j$, pairwise inequivalent, so
Lemma~\ref{L3} applies with $\sigma_E(P_J)=\sum_jP_J(j)\ket j\!\bra j$.  We
evaluate the exponent for the point-mass specialization $P_1=\delta_0$, for
which $\rho_{E|P_1}=\rho_E$. Sums below run over the
support of $q$, and $P_J$ may be taken supported there as well.  For
$t>0$,
\begin{equation}
\Tr\Bigl\{\Bigl(\rho_E^{\frac12}\sigma_E(P_J)^{\frac{-t}{1+t}}\rho_E^{\frac12}\Bigr)^{1+t}\Bigr\}
=\bra{\psi_E}\sigma_E(P_J)^{\frac{-t}{1+t}}\ket{\psi_E}^{\,1+t}
=\Bigl(\sum_jq_jP_J(j)^{\frac{-t}{1+t}}\Bigr)^{1+t},
\end{equation}
because $\rho_E$ has rank one.  Put
$s=\frac{1+t}{1+2t}\in(0,1)$ and $r=\frac{s}{s-1}=-\frac{1+t}{t}$.  The
reverse H\"older inequality \cite[Section~2.3.4]{Group2} gives
\begin{align}
\sum_jq_jP_J(j)^{\frac{-t}{1+t}}
&\ \ge\ \Bigl(\sum_jq_j^{s}\Bigr)^{\frac1s}\Bigl(\sum_j\bigl(P_J(j)^{\frac{-t}{1+t}}\bigr)^{r}\Bigr)^{\frac1r}
\notag\\
&\ =\ \Bigl(\sum_jq_j^{\frac{1+t}{1+2t}}\Bigr)^{\frac{1+2t}{1+t}}\Bigl(\sum_jP_J(j)\Bigr)^{-\frac{t}{1+t}}
=\Bigl(\sum_jq_j^{\frac{1+t}{1+2t}}\Bigr)^{\frac{1+2t}{1+t}},
\end{align}
with equality when $P_J(j)\propto q_j^{(1+t)/(1+2t)}$.  Hence
\begin{equation}
\min_{P_J}\Tr\Bigl\{\Bigl(\rho_E^{\frac12}\sigma_E(P_J)^{\frac{-t}{1+t}}\rho_E^{\frac12}\Bigr)^{1+t}\Bigr\}
=\Bigl(\sum_jq_j^{\frac{1+t}{1+2t}}\Bigr)^{1+2t}
=\Bigl(\sum_jq_j^{\frac{\alpha}{2\alpha-1}}\Bigr)^{2\alpha-1},\qquad\alpha=1+t,
\end{equation}
and \eqref{EE4} becomes, for a coding blocklength $N$,
\begin{equation}
\bbE_F\,\delta^E_N\ \le\ \min_{0\le t\le1}2^{\frac{1-t}{1+t}}
\exp\Bigl[-\frac{t}{1+t}N\rho_{\mathrm{sac},N}\Bigr]
\Bigl(\sum_jq_j^{\frac{1+t}{1+2t}}\Bigr)^{\frac{N(1+2t)}{1+t}}.
\Label{EE5}
\end{equation}
At $t=0$ the right-hand side is understood as its continuous extension,
which equals $2$; the useful range is $t>0$.

\subsection{Quantum BPSK}\Label{subsec:qbpsk}

Let $\ket\alpha=\sum_{n\ge0}e^{-|\alpha|^2/2}\alpha^n(n!)^{-1/2}\ket n$ be a
coherent state, let $G=\bbZ_2$, and take the seeds
$\rho_B=\ket\alpha\!\bra\alpha$ and $\rho_E=\ket{\alpha'}\!\bra{\alpha'}$ with
the representations $U_{B,0}=U_{E,0}=I$ and
$U_{B,1}=U_{E,1}=\sum_n(-1)^n\ket n\!\bra n$, the photon-number parity.
Then
\begin{equation}
W_{B,0}=\ket\alpha\!\bra\alpha,\quad W_{B,1}=\ket{-\alpha}\!\bra{-\alpha},\qquad
W_{E,0}=\ket{\alpha'}\!\bra{\alpha'},\quad W_{E,1}=\ket{-\alpha'}\!\bra{-\alpha'},
\Label{BNT}
\end{equation}
which is binary phase-shift keying (BPSK) seen by two receivers.  When
$|\alpha'|>|\alpha|$ Eve's states are more distinguishable than Bob's and
the private capacity is zero; when $|\alpha'|<|\alpha|$ the channel is
degraded and the private capacity is additive.  Indeed, when $|\alpha'|<|\alpha|$, a pure-loss channel of
transmissivity $|\alpha'/\alpha|^2$, followed by a phase rotation, maps
$\ket{\pm\alpha}$ to $\ket{\pm\alpha'}$; when the inequality is
reversed the same construction maps Eve's states to Bob's.
Neither case is a candidate
for nonadditivity.  The hybrid channel below is what one obtains from BPSK
when Bob commits to a fixed measurement; see Section~\ref{sec:bpsk}.

\subsection{Hybrid binary example: a symmetric channel to Bob and pure states to Eve}\Label{subsec:hybrid_binary}

Let $G=\bbZ_2$ and take the seeds
\begin{equation}
\rho_B=\begin{pmatrix}1-q&0\\0&q\end{pmatrix},\qquad
\rho_E=\begin{pmatrix}q'&\sqrt{q'(1-q')}\\\sqrt{q'(1-q')}&1-q'\end{pmatrix}
=\ket{\psi_E}\!\bra{\psi_E},
\end{equation}
where $\ket{\psi_E}=\sqrt{q'}\ket0+\sqrt{1-q'}\ket1$,
with the representations
\begin{equation}
U_{B,0}=U_{E,0}=I,\qquad
U_{B,1}=\begin{pmatrix}0&1\\1&0\end{pmatrix},\qquad
U_{E,1}=\begin{pmatrix}1&0\\0&-1\end{pmatrix}.
\end{equation}
Bob therefore sees a binary symmetric channel with crossover probability
$q$, and Eve receives one of the two pure states
$\ket{\psi_x}=\sqrt{q'}\ket0+(-1)^x\sqrt{1-q'}\ket1$, whose overlap is
$\langle\psi_0|\psi_1\rangle=2q'-1$, of magnitude $1-2q'$ for $q'<1/2$.  For
the input distribution $\Bern(x)$, with $x$ the probability of the symbol
$1$, Bob's output is $\Bern(x\star q)$ with $x\star q:=x(1-q)+(1-x)q$, and
Eve's state is
\begin{equation}
(1-x)W_{E,0}+xW_{E,1}=\begin{pmatrix}q'&(1-2x)\sqrt{q'(1-q')}\\(1-2x)\sqrt{q'(1-q')}&1-q'\end{pmatrix},
\end{equation}
with eigenvalues $\tfrac12\bigl(1\pm\sqrt{1-16q'(1-q')x(1-x)}\bigr)$.  Writing
$f_1(x,q):=h(x\star q)$ and
$f_2(x,q'):=h\bigl(\tfrac12\bigl(1-\sqrt{1-16q'(1-q')x(1-x)}\bigr)\bigr)$,
the single-use entropy difference is
\begin{equation}
\DeltaBE(x;q,q')=f_1(x,q)-f_2(x,q'),
\Label{eq:single_use_Delta}
\end{equation}
a function symmetric about $x=1/2$.  (The off-diagonal coefficient of
Eve's state is $2x-1$ if one writes $xW_{E,0}+(1-x)W_{E,1}$ instead; the two
matrices are conjugate by $U_{E,1}$ and have the same entropy.)

Theorem~\ref{thm:main} asks for the maximum and the minimum of
\eqref{eq:single_use_Delta} over $x\in[0,1/2]$.  For $q=0.05$ and
$q'=0.15$ the minimum is $\simeq0.191854$, attained at $x\simeq0.0162174$,
and the maximum is $\simeq0.270438$, attained at the uniform input
$x=1/2$; the latter was verified numerically on a fine grid together with
the sign of the derivative below, and it is consistent with
$\DeltaBE(0)=h(0.05)\simeq0.198515<0.270438$.  The equality condition of
Theorem~\ref{thm:main} therefore applies, and
\begin{equation}
C^{(1)}(W_B,W_E)\simeq0.078584\ \text{nats}.
\end{equation}
Figure~\ref{1copy} shows the function.  Its minimum is attained at a
nonzero $x$, which is the single-use face of noisy preprocessing: the best
noise distribution $P_1$ in the uniform bound \eqref{eq:uniform_bound} is
not a point mass.  For $q=0.05$ and $q'=0.12$ the minimum is
$\simeq0.196473$ at $x\simeq0.0052765$, the maximum $\simeq0.326222$ is
again at $x=1/2$, and $C^{(1)}\simeq0.129749$ nats.

\begin{figure}[t]
\centering
\begin{tikzpicture}[x=11cm,y=38cm]
\draw[->] (0,0.185) -- (0.53,0.185) node[right] {$x$};
\draw[->] (0,0.185) -- (0,0.282) node[above] {$f_1(x,0.05)-f_2(x,0.15)$};
\foreach \x/\lab in {0/0,0.25/0.25,0.5/0.5}{\draw (\x,0.185) -- (\x,0.1835) node[below] {\scriptsize \lab};}
\foreach \y/\lab in {0.19/0.19,0.23/0.23,0.27/0.27}{\draw (0,\y) -- (-0.006,\y) node[left] {\scriptsize \lab};}
\draw[thick] plot[smooth] coordinates {
(0.000,0.198515)
(0.020,0.191983)
(0.040,0.195210)
(0.060,0.200286)
(0.080,0.205940)
(0.100,0.211706)
(0.120,0.217374)
(0.140,0.222840)
(0.160,0.228047)
(0.180,0.232967)
(0.200,0.237585)
(0.220,0.241894)
(0.240,0.245891)
(0.260,0.249576)
(0.280,0.252952)
(0.300,0.256020)
(0.320,0.258785)
(0.340,0.261249)
(0.360,0.263415)
(0.380,0.265286)
(0.400,0.266865)
(0.420,0.268154)
(0.440,0.269154)
(0.460,0.269868)
(0.480,0.270296)
(0.500,0.270438)
};
\fill (0.0162,0.191854) circle[radius=1.1pt];
\node[align=left, anchor=west] at (0.05,0.1925) {\scriptsize minimum\\[-1mm]\scriptsize $(0.0162,0.191854)$};
\end{tikzpicture}
\caption{The single-use entropy difference $f_1(x,0.05)-f_2(x,0.15)$ of the hybrid binary channel as a function of the input probability $x$.  The minimum sits at a nonzero $x$, so the optimal sender noise is not a point mass: noisy preprocessing already pays off at a single channel use.}
\Label{1copy}
\end{figure}
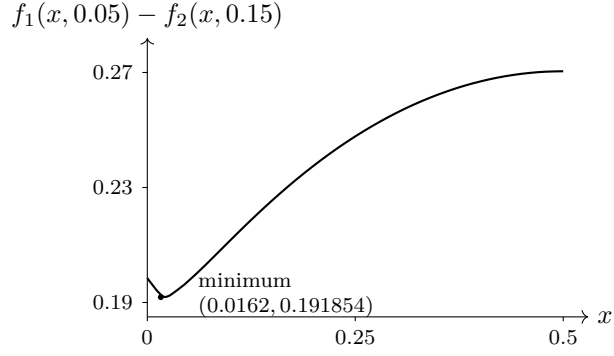

\medskip
\noindent\textit{Derivative.}  For numerical work it is convenient to have
the derivative of \eqref{eq:single_use_Delta} on $x\in(0,1/2)$.  With
$a(x):=x\star q$, $s(x):=\sqrt{1-16q'(1-q')x(1-x)}$ and $b(x):=(1-s(x))/2$,
\begin{equation}
\frac{d}{dx}\DeltaBE(x;q,q')
=(1-2q)\log\frac{1-a(x)}{a(x)}
-\frac{4q'(1-q')(1-2x)}{s(x)}\log\frac{1+s(x)}{1-s(x)}.
\Label{eq:derivative}
\end{equation}
The stationary points reported above are the zeros of this expression.

\subsection{The two-copy functional}\Label{subsec:two_copy_binary}

For a distribution $P$ on $\bbZ_2^2$ the two-copy entropy difference is
\begin{equation}
\DeltaBE_2(P;q,q'):=H\Bigl(\sum_{g\in\bbZ_2^2}P(g)\,W^B_{g_1}\otimes W^B_{g_2}\Bigr)
-H\Bigl(\sum_{g\in\bbZ_2^2}P(g)\,W^E_{g_1}\otimes W^E_{g_2}\Bigr).
\end{equation}
Bob's part is classical: his output distribution is the XOR convolution
of $P$ with two independent $\Bern(q)$ flips,
\begin{equation}
P_B(y)=\sum_{g\in\bbZ_2^2}P(g)\,q^{|y\oplus g|}(1-q)^{2-|y\oplus g|},\qquad y\in\bbZ_2^2,
\end{equation}
where $|\cdot|$ is the Hamming weight, and the first entropy is its Shannon
entropy.  Eve's part has a transparent structure.  In the computational
basis, $\langle j|\psi_x\rangle=\sqrt{w_j}\,(-1)^{x\cdot j}$ with
$w_j=q'^{\,2-|j|}(1-q')^{|j|}$ for $j\in\bbZ_2^2$, so
\begin{equation}
\sum_gP(g)\,\ket{\psi_g}\!\bra{\psi_g}=D\,\bigl[c_P(j\oplus k)\bigr]_{j,k\in\bbZ_2^2}\,D,
\qquad
D=\operatorname{diag}(\sqrt{w_j}),
\Label{eq:eve_two_copy_structure}
\end{equation}
where
\begin{equation}
c_P(d):=\sum_{g\in\bbZ_2^2}P(g)(-1)^{g\cdot d},\qquad d\in\bbZ_2^2.
\end{equation}
The numbers $c_P(d)$ are the Fourier coefficients of $P$ on $\bbZ_2^2$, with
$c_P(00)=1$.  Written out in the basis order $(00,01,10,11)$ and with
$s:=\sqrt{q'(1-q')}$, Eve's state is
\begin{equation}
\begin{pmatrix}
q'^2 & c_P(01)\,q's & c_P(10)\,q's & c_P(11)\,q'(1-q')\\
c_P(01)\,q's & q'(1-q') & c_P(11)\,q'(1-q') & c_P(10)\,(1-q')s\\
c_P(10)\,q's & c_P(11)\,q'(1-q') & q'(1-q') & c_P(01)\,(1-q')s\\
c_P(11)\,q'(1-q') & c_P(10)\,(1-q')s & c_P(01)\,(1-q')s & (1-q')^2
\end{pmatrix},
\Label{NM2}
\end{equation}
and the second entropy in $\DeltaBE_2$ is the entropy of its spectrum.  In
terms of the parametrization of Table~\ref{tab:two_copy_params},
$c_P(10)=t_1-t_2+t_3$, $c_P(01)=t_1-t_2-t_3$ and $c_P(11)=2t_1+2t_2-1$.
The same formula \eqref{eq:eve_two_copy_structure} holds for $n$ copies with
$j,k,d\in\bbZ_2^n$ and $w_j=q'^{\,n-|j|}(1-q')^{|j|}$, and is used to evaluate Eve's entropy in the numerical comparison
and in Section~\ref{sec:witness}.

\begin{table}[t]
\centering
\caption{Parametrization of a distribution on $\bbZ_2^2$.  The range of $(t_1,t_2,t_3)$ is $t_1,t_2\ge0$ and $|t_3|\le1-t_1-t_2$.}
\Label{tab:two_copy_params}
\begin{tabular}{cc}
\toprule
Input pair & Probability \\
\midrule
$(0,0)$ & $t_1$ \\
$(1,1)$ & $t_2$ \\
$(0,1)$ & $(1-t_1-t_2+t_3)/2$ \\
$(1,0)$ & $(1-t_1-t_2-t_3)/2$ \\
\bottomrule
\end{tabular}
\end{table}

\subsection{The pair family at two uses and the repetition subfamily}\Label{subsec:pair_two_uses}

Corollary~\ref{cor:finite_block_covariant} with $G^2=\bbZ_2^2$ defines
\begin{equation}
R_2^{\pair}(q,q')=\frac12\max_{P_1,P_3\in\calP(\bbZ_2^2)}
\bigl[\DeltaBE_2(P_3*P_1;q,q')-\DeltaBE_2(P_1;q,q')\bigr],
\end{equation}
the factor $\tfrac12$ converting the two-use value into a rate per channel
use.  Both distributions are optimization variables; fixing $P_3$ to the
uniform distribution on all of $\bbZ_2^2$ would give a valid but needlessly
restricted bound.  For a feasible pair found numerically we write
\begin{equation}
R_2^{\ach}(q,q'):=\frac12\bigl[\DeltaBE_2(P_3*P_1;q,q')-\DeltaBE_2(P_1;q,q')\bigr]\ \le\ R_2^{\pair}(q,q'),
\end{equation}
the achieved rate of the reported pair.

The noisy repetition codes of \cite{SRS,TBR} sit inside this family.  Let
$P_3^{\rep}$ be the uniform distribution on the repetition subgroup
$\{00,11\}$ (not on the whole group) and let
$P_{1,\lambda}=\Bern(\lambda)^{\otimes2}$, that is,
$P_{1,\lambda}(00)=(1-\lambda)^2$, $P_{1,\lambda}(01)=P_{1,\lambda}(10)=\lambda(1-\lambda)$,
$P_{1,\lambda}(11)=\lambda^2$.  The pair $(P_{1,\lambda},P_3^{\rep})$ is the
two-use repetition code with a uniform logical bit and independent physical
flips of probability $\lambda$: the outer variable $V\in\{00,11\}$ is the
repeated logical bit and $Z\sim P_{1,\lambda}$ is the physical noise.  A
slightly larger subfamily is obtained by also flipping the logical bit with
probability $\mu$ before it is repeated.  If $U$ is the uniform logical bit,
$L\sim\Bern(\mu)$ is independent of it and $T=U\oplus L$ is repeated and
then subjected to the physical flips, the conditional input distribution
given $U=0$ is
\begin{equation}
P_{1,\lambda,\mu}:=(1-\mu)\Bern(\lambda)^{\otimes2}+\mu\,\Bern(1-\lambda)^{\otimes2},
\Label{eq:logical_flip_P1}
\end{equation}
and the one given $U=1$ is its translate by $11$.  So the logically
prefixed repetition code is the pair $(P_{1,\lambda,\mu},P_3^{\rep})$, and
its rate is
\begin{equation}
r_{2,\lambda,\mu}(q,q')=\tfrac12\bigl[\DeltaBE_2(P_3^{\rep}*P_{1,\lambda,\mu})-\DeltaBE_2(P_{1,\lambda,\mu})\bigr].
\end{equation}
The same construction with $\Bern(\lambda)^{\otimes n}$, $\Bern(1-\lambda)^{\otimes n}$
and $P_3^{\rep}$ uniform on $\{0^n,1^n\}$ gives the $n$-use version
$r_{n,\lambda,\mu}$.

For exchangeable two-use noise $P_1=(a,b,b,c)$, the parameters in
\eqref{eq:logical_flip_P1} satisfy
\[
b=\lambda(1-\lambda),\qquad a-c=(1-2\mu)(1-2\lambda).
\]
Such a representation exists when $b\le1/4$ and
$|a-c|\le\sqrt{1-4b}$, allowing a fixed complement of the input labels.
These identities identify the repetition parameters of the pairs found
numerically. They do not cover all exchangeable distributions or all
outer distributions. At larger blocklengths we optimize the full
Hamming-weight profile of exchangeable noise, rather than restricting it
to complementary binomial components.

\section{Repetition-code benchmarks and numerical comparison}\Label{sec:TBR}

\subsection{The Tikku--Berta--Renes line and three repetition benchmarks}\Label{subsec:benchmarks}

Tikku, Berta and Renes \cite{TBR} study the classical-input wiretap channel
\begin{equation}
W_{Y|X}[p]=\mathrm{BSC}(p),\qquad W_{C|X}[p]=\mathrm{BPC}(1-2p),
\end{equation}
a binary symmetric channel to Bob and a binary pure-state channel to Eve
whose two states have overlap of magnitude $1-2p$.  In the notation of
Section~\ref{subsec:hybrid_binary} this is the line
\begin{equation}
q=p,\qquad q'=p,
\end{equation}
and all computations in this section are on this line, for
$p\in[0.121,0.125]$, on the fifty-point grid $p_k=(5929+4k)/49000$,
$k=0,\dots,49$.  The interval is chosen because the single-use private
information vanishes inside it: \cite{TBR} show that $C^{(1)}=0$ for
$p\ge\hat p\approx0.1241$, and our own evaluation of
Lemma~\ref{lem:binary_single_use} puts $\hat p\approx0.12412$.

The repetition code with noisy preprocessing, $\Rep_{n,\lambda}$, maps a
logical bit $u$ to $u^n$ and then flips each of the $n$ physical bits
independently with probability $\lambda$. We distinguish three rates.

The first takes the logical bit uniform and optimizes only $\lambda$:
\begin{equation}
R_n^{\rep,0}(p):=\max_{0\le\lambda\le1/2}\frac1n\bigl[I(U;B^n)-I(U;E^n)\bigr]_{U\ \text{uniform}}
=\max_{0\le\lambda\le1/2}\frac1n\bigl[B_n(\tilde p_\lambda)-E_{n,\lambda}(p)\bigr],
\Label{eq:R_rep0}
\end{equation}
where $\tilde p_\lambda=p(1-\lambda)+(1-p)\lambda$ is Bob's effective
crossover probability,
\begin{equation}
B_n(\tilde p)=\log2-\sum_{s\in\{0,1\}^{n-1}}P^{(n)}_{\tilde p}(s)\,H\bigl(P^{(n)}_{\tilde p}(U|s)\bigr)
\end{equation}
is the mutual information of the repetition code over $\mathrm{BSC}(\tilde p)$
written through the syndrome $s$ (with $P^{(n)}_{\tilde p}(0,s)=\tilde p^{\,|s|}(1-\tilde p)^{n-|s|}$,
$P^{(n)}_{\tilde p}(1,s)=\tilde p^{\,n-|s|}(1-\tilde p)^{|s|}$), and
\begin{equation}
E_{n,\lambda}(p)=H\Bigl(\tfrac12\sigma_{0,\lambda}^{\otimes n}+\tfrac12\sigma_{1,\lambda}^{\otimes n}\Bigr)-nH(\sigma_{0,\lambda}),
\qquad \sigma_{0,\lambda}=(1-\lambda)\rho_0+\lambda\rho_1,\quad \sigma_{1,\lambda}=U_{E,1}\sigma_{0,\lambda}U_{E,1}^\dagger,
\end{equation}
is Eve's Holevo information, $\rho_0,\rho_1$ being her two pure states.
This is the uniform-input expression of \cite{TBR}, and in the language of
Section~\ref{subsec:pair_two_uses} it is the pair $(P_{1,\lambda},P_3^{\rep})$.

The second also optimizes a logical flip $\mu$ before the repetition, as in
\eqref{eq:logical_flip_P1}:
\begin{equation}
R_n^{\rep,\log}(p):=\max_{\lambda,\mu\in[0,1/2]}r_{n,\lambda,\mu}(p),
\Label{eq:R_replog}
\end{equation}
where
\begin{align}
r_{n,\lambda,\mu}=\frac1n\Bigl[&H\Bigl(\tfrac{\omega^B_0+\omega^B_1}{2}\Bigr)-H\bigl((1-\mu)\omega^B_0+\mu\omega^B_1\bigr)\notag\\
&-H\Bigl(\tfrac{\omega^E_0+\omega^E_1}{2}\Bigr)+H\bigl((1-\mu)\omega^E_0+\mu\omega^E_1\bigr)\Bigr]
\end{align}
and $\omega^B_u,\omega^E_u$ are the outputs of
$W^{\otimes n}\circ\Rep_{n,\lambda}$ for logical input $u$; the two
conditional entropies for $U=0,1$ agree by covariance, which is why only
one appears.

The third is the full noisy-repetition benchmark of \cite{TBR}:
\begin{equation}
R_n^{\rep}(p):=\frac1n\max_{0\le\lambda\le1/2}P_1\bigl(W^{\otimes n}\circ\Rep_{n,\lambda}\bigr),
\Label{eq:R_rep_full}
\end{equation}
the full single-letter private information of the repetition-preprocessed
channel, in which the logical bit may be preprocessed arbitrarily, with
an auxiliary variable of any size.  Since the preprocessed channel has a
binary input, Lemma~\ref{lem:binary_single_use} gives it explicitly:
with $F_\lambda(x):=\DeltaBE_n\bigl(x\Bern(1-\lambda)^{\otimes n}+(1-x)\Bern(\lambda)^{\otimes n}\bigr)$
the entropy difference as a function of the logical prior,
\begin{equation}
R_n^{\rep}(p)=\frac1n\max_{\lambda}\max_{x\in[0,1]}\bigl[F_\lambda(x)-F_\lambda^{\mathrm{vex}}(x)\bigr].
\Label{eq:R_rep_envelope}
\end{equation}
The three benchmarks are nested,
\begin{equation}
R_n^{\rep}\ \ge\ R_n^{\rep,\log}\ \ge\ R_n^{\rep,0},
\Label{eq:benchmark_nesting}
\end{equation}
because the logical flip is one particular two-point decomposition of the
uniform prior and no flip is one particular logical flip.  All three are
achievable rates.  The first two constructions belong to the covariant pair
family; the third does not in general, because its conditional input
distributions $x_v\Bern(1-\lambda)^{\otimes n}+(1-x_v)\Bern(\lambda)^{\otimes n}$
need not be translates of one another, but whenever the optimal
decomposition in \eqref{eq:R_rep_envelope} is the symmetric two-point one,
which is what we find below, its value is attained inside the family.

Optimizing the logical flip matters near the threshold. At
$p=6081/49000\approx0.12410$,
$R_3^{\rep,0}\simeq1.40\times10^{-7}$ but $R_3^{\rep,\log}\simeq2.57\times10^{-5}$,
a factor of about $180$. We therefore compare against the optimized
benchmark, including logical preprocessing.

We evaluated \eqref{eq:R_rep_envelope} at every grid point and every
$n\in\{2,3,4\}$ by sampling the logical prior and computing the lower
convex hull (Appendix~\ref{app:numerics}).  At every positive point the maximizing
decomposition was the symmetric two-point one at $x=1/2$: numerical
evaluation of the full repetition benchmark agrees with the optimized
logical-flip construction on the tested grid, to within
$7.1\times10^{-9}$ nats per use, consistent with the discretization of the
prior.  The figures and tables below show $R_n^{\rep,\log}$, computed by
multi-start local optimization over $(\lambda,\mu)$, and we use these
values as numerical estimates of the full benchmark $R_n^{\rep}$.

\subsection{The single-use reference curve}\Label{subsec:single_use_reference}

The natural single-use reference is the uniform bound
\eqref{eq:uniform_bound},
\begin{equation}
R_1(p):=\DeltaBE(\tfrac12;p,p)-\min_{0\le x\le1}\DeltaBE(x;p,p),
\end{equation}
an achievable rate for every $p$.  On this line the equality condition of
Theorem~\ref{thm:main} fails: $\DeltaBE(0)=h(p)$, since a deterministic
input gives Bob a $\Bern(p)$ output and Eve a pure state, while
$\DeltaBE(\tfrac12)=\log2-h(p)$, and $h(p)>\log2-h(p)$ throughout
$[0.121,0.125]$.  So the sufficient condition of Theorem~\ref{thm:main}
does not apply on this interval, and $R_1$ is not identified with
$C^{(1)}$ by the theorem.  Lemma~\ref{lem:binary_single_use}
settles the matter directly.  Numerically, $\DeltaBE(\cdot;p,p)$ has two
symmetric global minima at $x^*$ and $1-x^*$, is convex outside
$[x^*,1-x^*]$ and unimodal inside, so its lower convex envelope is
$\DeltaBE$ itself outside and the constant $\min\DeltaBE$ inside, and the gap
$\DeltaBE-\DeltaBE^{\mathrm{vex}}$ is largest at $x=1/2$, where it equals
$R_1$.  We verified this on the grid: at the 39 points with $C^{(1)}>0$
the maximizing $x$ in \eqref{eq:binary_single_use} is $1/2$ and
$|C^{(1)}-R_1|<9.8\times10^{-12}$, and at the remaining points,
$p\ge6085/49000$, both computed values vanish to numerical precision.
The curve $R_1$ in Fig.~\ref{fig:rates} may thus be read as $C^{(1)}$,
with the understanding that this identification is numerical except at the
parameter of Theorem~\ref{thm:witness}, where $C^{(1)}=0$ is proved.

\subsection{Two uses: agreement with logically prefixed repetition}\Label{subsec:two_use_results}

For $n=2$ we searched the full pair family $(P_1,P_3)\in\calP(\bbZ_2^2)^2$
at every grid point, combining differential evolution, exchangeable-noise
searches and full-pair local refinement
(Appendix~\ref{app:numerics}). No improvement over logically prefixed
repetition exceeded the numerical tie tolerance. The retained pairs have
$P_3$ uniform on $\{00,11\}$ and $P_1$ of the form
\eqref{eq:logical_flip_P1}. At $p=6081/49000$, for instance, the
retained pair has
\[
P_1\simeq(0.36372,\,0.21717,\,0.21717,\,0.20194),\qquad P_3=(\tfrac12,0,0,\tfrac12),
\qquad R_2^{\ach}\simeq7.5961\times10^{-6},
\]
and matching $b=\lambda(1-\lambda)$ and $a-c=(1-2\mu)(1-2\lambda)$ gives
$\lambda\simeq0.31880$, $\mu\simeq0.27679$; the uniform-logical-input value
at the same point is $R_2^{\rep,0}\simeq6.85\times10^{-8}$.  The whole gap
between the two is the logical flip.  Table~\ref{tab:two_use} lists the
values at a few grid points, and panel (a) of Fig.~\ref{fig:rates} shows the
curves; the achieved two-use curve lies on the benchmark.

\subsection{Three and four uses: noise beyond physical and logical flips}\Label{subsec:three_four_use_results}

For $n=3,4$, the retained pairs have $P_3$ uniform on
$\{0^n,1^n\}$ and exchangeable noise $P_1$. Optimizing its Hamming-weight
profile gives rates above the complementary-binomial family
\eqref{eq:logical_flip_P1}.

At $p=6081/49000$, the three-use profile is approximately
\[
\bigl(P_1(|z|=0),P_1(|z|=1),P_1(|z|=2),P_1(|z|=3)\bigr)
\simeq(0.278642,0.386189,0.248223,0.086946).
\]
Each word of weight $j$ receives the corresponding mass divided by
$\binom3j$. This particular profile admits the representation
\[
0.716526\,\mathrm{Bin}(3,0.282895)
+0.283474\,\mathrm{Bin}(3,0.629531)
\]
to the displayed precision. Operationally, flip the logical bit with
probability $0.283474$, and use physical flip probability $0.282895$ if
it was not flipped and $0.370469$ if it was. The standard construction
requires these two physical noise levels to be equal. The achieved rate
is $R_3^{\ach}\simeq2.7574\times10^{-5}$, compared with
$R_3^{\rep,\log}\simeq2.5672\times10^{-5}$.

At four uses the rates are $R_4^{\ach}\simeq5.7400\times10^{-5}$ and
$R_4^{\rep,\log}\simeq4.8913\times10^{-5}$. A simple rounded noise
profile attaining the displayed achieved rate is
\begin{equation}
w=(0.23089,0.33839,0.27248,0.11757,0.04067).
\Label{eq:four_use_profile}
\end{equation}
This profile cannot be a mixture of two binomial distributions. To see
this, let $K=|Z|$ and define
$m_j=\mathbb E[\binom Kj]/\binom4j$ for $0\le j\le4$.
For a mixture of two binomials, $m_j=a\theta^j+(1-a)\varphi^j$, so the
moment matrix $(m_{i+j})_{i,j=0}^2$ has rank at most two. For the exact
decimal distribution \eqref{eq:four_use_profile}, direct rational
evaluation gives
\[
\det(m_{i+j})_{i,j=0}^2>6.8\times10^{-5}>0.
\]
Thus the four-use example uses noise beyond this two-component description.

Across the grid the picture is uniform (Fig.~\ref{fig:rates}(b) and
Fig.~\ref{fig:gains}).  The achieved three-use rate exceeds our numerical estimate of the three-use
repetition benchmark at all 44 grid points at which that estimate
is positive; at the remaining six grid points, the retained rates are zero to
numerical precision.  The absolute gain reaches $7.76\times10^{-6}$ nats per channel use at
$p\simeq0.1226$, and the relative gain grows from about $0.3\%$
at $p=0.121$ to about $24\%$ at $p\simeq0.1244$, the last grid
point at which the benchmark still exceeds $10^{-6}$ nats; beyond that
point both rates are tiny and ratios are not meaningful.  The four-use
rate exceeds our numerical estimate of the four-use benchmark at all 46 points at which the
estimate is positive, with an absolute gain of up to $2.06\times10^{-5}$ nats per
channel use at $p\simeq0.1226$ and a relative gain from about
$1.0\%$ to about $129\%$ at $p\simeq0.1246$.
Table~\ref{tab:three_four_use} gives the values at the same selected grid
points as Table~\ref{tab:two_use}.

\begin{figure}[t]
\centering
\includegraphics[width=0.62\linewidth]{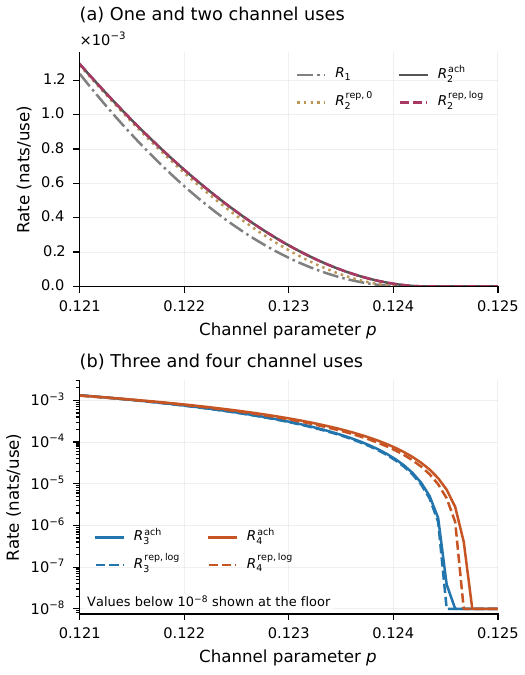}
\caption{Achieved rates and repetition benchmarks on the line $q=q'=p$.
(a) One and two channel uses: the single-use reference $R_1$, the
uniform-logical-input repetition value $R_2^{\rep,0}$, the repetition
benchmark $R_2^{\rep,\log}$ with optimized physical and logical flips, and
the achieved two-use pair value $R_2^{\ach}$, which coincides with
$R_2^{\rep,\log}$.  (b) Three and four channel uses on a logarithmic scale:
the achieved pair values $R_n^{\ach}$ (solid) and the repetition benchmarks
$R_n^{\rep,\log}$ (dashed), with values below $10^{-8}$ drawn at the floor.
The benchmark curves are the numerical estimates of the full benchmark
$R_n^{\rep}$ described in Section~\ref{subsec:benchmarks}.  All curves are
achieved values of explicit distributions, joined by straight lines between
the fifty grid points.}
\Label{fig:rates}
\end{figure}

\begin{figure}[t]
\centering
\includegraphics[width=0.62\linewidth]{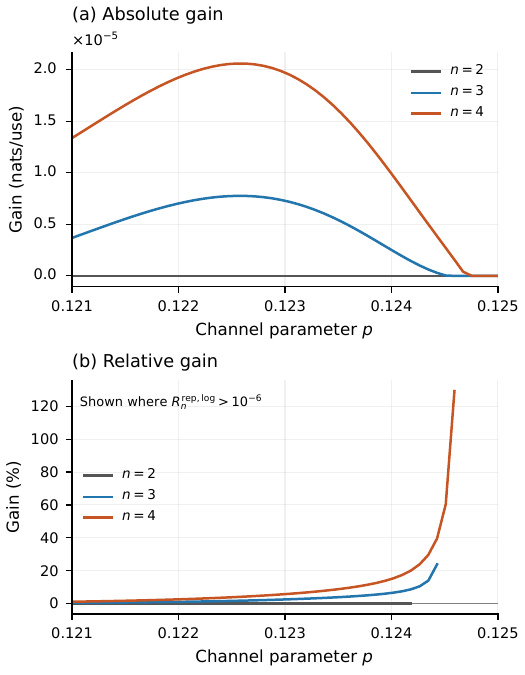}
\caption{Gain of the achieved covariant pair value over the numerical estimate of the repetition
benchmark at the same blocklength, $R_n^{\ach}-R_n^{\rep,\log}$, for $n=2,3,4$:
(a) in nats per channel use, (b) as a percentage of the benchmark, shown
where the benchmark exceeds $10^{-6}$ nats.  The two-use gain is zero to
numerical precision.}
\Label{fig:gains}
\end{figure}

\begin{table}[!htbp]
\centering
\caption{One and two channel uses at selected grid points (nats per channel use, four significant digits).  The last columns give the parameters of the logically prefixed repetition code with which the two-use pair search agrees.}
\Label{tab:two_use}
\begin{tabular}{lcccccc}
\toprule
$p$ & $R_1$ & $R_2^{\rep,0}$ & $R_2^{\rep,\log}$ & $R_2^{\ach}$ & $\lambda$ & $\mu$\\
\midrule
0.121000 & $1.238\times10^{-3}$ & $1.289\times10^{-3}$ & $1.295\times10^{-3}$ & $1.295\times10^{-3}$ & 0.2335 & 0.0148\\
0.121980 & $5.950\times10^{-4}$ & $6.728\times10^{-4}$ & $6.902\times10^{-4}$ & $6.902\times10^{-4}$ & 0.2533 & 0.0374\\
0.122959 & $1.789\times10^{-4}$ & $2.251\times10^{-4}$ & $2.529\times10^{-4}$ & $2.529\times10^{-4}$ & 0.2783 & 0.0843\\
0.123612 & $3.475\times10^{-5}$ & $4.806\times10^{-5}$ & $7.308\times10^{-5}$ & $7.308\times10^{-5}$ & 0.2994 & 0.1534\\
0.124102 & $4.514\times10^{-8}$ & $6.853\times10^{-8}$ & $7.596\times10^{-6}$ & $7.596\times10^{-6}$ & 0.3188 & 0.2768\\
0.124265 & $0$ & $0$ & $5.857\times10^{-7}$ & $5.857\times10^{-7}$ & 0.3262 & 0.3763\\
\bottomrule
\end{tabular}
\end{table}

\begin{table}[!htbp]
\centering
\caption{Three and four channel uses at the same grid points: repetition benchmark, achieved covariant pair value and relative gain.}
\Label{tab:three_four_use}
\begin{tabular}{lcccccc}
\toprule
$p$ & $R_3^{\rep,\log}$ & $R_3^{\ach}$ & gain & $R_4^{\rep,\log}$ & $R_4^{\ach}$ & gain\\
\midrule
0.121000 & $1.320\times10^{-3}$ & $1.324\times10^{-3}$ & 0.3\% & $1.296\times10^{-3}$ & $1.309\times10^{-3}$ & 1.0\%\\
0.121980 & $7.558\times10^{-4}$ & $7.628\times10^{-4}$ & 0.9\% & $7.861\times10^{-4}$ & $8.053\times10^{-4}$ & 2.4\%\\
0.122959 & $3.164\times10^{-4}$ & $3.237\times10^{-4}$ & 2.3\% & $3.643\times10^{-4}$ & $3.842\times10^{-4}$ & 5.5\%\\
0.123612 & $1.150\times10^{-4}$ & $1.199\times10^{-4}$ & 4.2\% & $1.545\times10^{-4}$ & $1.694\times10^{-4}$ & 9.6\%\\
0.124102 & $2.567\times10^{-5}$ & $2.757\times10^{-5}$ & 7.4\% & $4.891\times10^{-5}$ & $5.740\times10^{-5}$ & 17.3\%\\
0.124265 & $9.642\times10^{-6}$ & $1.065\times10^{-5}$ & 10.4\% & $2.585\times10^{-5}$ & $3.201\times10^{-5}$ & 23.8\%\\
\bottomrule
\end{tabular}
\end{table}

For each optimized pair, we checked the achieved rate by forming the
physical tensor-product output states directly. The gains also persist across blocklengths. Among the
repetition codes of length two to four, length three is best at the first
five grid points ($p\le0.1213$) and length four beyond; the three-use
pair exceeds the computed four-use repetition value at the first six grid
points ($p\le0.1214$), and the four-use pair exceeds the numerically computed envelope
$\max_{2\le n\le4}R_n^{\rep,\log}$ at 44 points.  Taken together,
the best pair found at blocklength at most four exceeds this numerical envelope at
every grid point where the envelope is positive.  Lengthening the block and
optimizing its internal structure are complementary routes to higher
rates.

\FloatBarrier
\section{An explicit nonadditivity witness}\Label{sec:witness}

We now prove nonadditivity at the explicit parameter $p=621/5000$.
Convexity of the single-use entropy difference gives the converse, and
exact rational entropy bounds certify a positive two-use rate. The
construction is an explicit instance of the phenomenon established
in \cite{TBR}.

\begin{samepage}
\begin{theorem}[Two-use nonadditivity at $p=621/5000$]\Label{thm:witness}
Consider the hybrid binary channel on the line $q=q'=p$ at $p=621/5000=0.1242$.
Then $C^{(1)}(W_B,W_E)=0$.  Moreover, for the pair
\begin{equation}
P_1=(0.364,\ 0.217,\ 0.217,\ 0.202),\qquad P_3=(0,\ \tfrac12,\ \tfrac12,\ 0)
\Label{eq:simple_witness}
\end{equation}
in the order $(00,01,10,11)$, the two-use rate
$R_2^{\ach}=\tfrac12[\DeltaBE_2(P_3*P_1)-\DeltaBE_2(P_1)]$ satisfies
\begin{equation}
10^{-6}\ <\ R_2^{\ach}\ <\ 1.1\times10^{-6}\quad\text{nats per channel use}.
\Label{eq:witness_interval}
\end{equation}
Consequently $C^{(2)}(W_B,W_E)\ge2R_2^{\ach}>0=2\,C^{(1)}(W_B,W_E)$.
\end{theorem}
\end{samepage}

The probabilities in \eqref{eq:simple_witness} have three decimal
places and are interpreted exactly. Its outer distribution is uniform
on the coset $\{01,10\}$, equivalent by translation to the repetition
subgroup. The noise has the form \eqref{eq:logical_flip_P1}, with
$\lambda\simeq0.31834$ and $\mu\simeq0.27706$. The witness therefore
belongs to the logically prefixed repetition family.

\subsection{The single-use converse: convexity of the entropy difference}

Fix $p\in(0,1/2)$ and put
\begin{equation}
a=1-2p,\qquad u=a^2,\qquad k=1-u,\qquad z=1-2x,\qquad t=z^2,\qquad r=\sqrt{u+kt},
\end{equation}
where $x$ is the probability of the symbol $1$.  On the line $q=q'=p$ the
single-use entropy difference \eqref{eq:single_use_Delta} reads, in the
variable $z\in[-1,1]$,
\begin{equation}
\DeltaBE(z)=h\Bigl(\frac{1+az}{2}\Bigr)-h\Bigl(\frac{1+r}{2}\Bigr),
\Label{eq:Delta_z}
\end{equation}
since Bob's output is $\Bern\bigl((1-az)/2\bigr)$ and Eve's eigenvalues are
$(1\pm r)/2$.  Because $z=1-2x$ is an affine change of variable, convexity of
\eqref{eq:Delta_z} in $z$ is the same as convexity of $\DeltaBE$ in $x$.

\begin{lemma}\Label{lem:second_derivative}
For $|z|<1$,
\begin{equation}
\DeltaBE''(z)=-\frac{u}{1-ut}+\frac{ku\,\atanh r}{r^3}+\frac{k^2t}{r^2(1-r^2)}.
\Label{eq:second_derivative}
\end{equation}
\end{lemma}

\begin{proof}
Using $h'(y)=\log\frac{1-y}{y}$, the first term of \eqref{eq:Delta_z} has
derivative $-a\atanh(az)$ and second derivative $-a^2/(1-a^2z^2)=-u/(1-ut)$.
For the second term, $\frac{d}{dz}h\bigl(\frac{1+r}{2}\bigr)=-\atanh(r)\,r'(z)$
with $r'(z)=kz/r$, so
\begin{align*}
-\frac{d^2}{dz^2}h\Bigl(\frac{1+r}{2}\Bigr)
&=k\,\frac{d}{dz}\Bigl[\frac{z\,\atanh r}{r}\Bigr]
=k\Bigl[\frac{\atanh r}{r}+kt\Bigl(\frac{1}{r^2(1-r^2)}-\frac{\atanh r}{r^3}\Bigr)\Bigr]\\
&=\frac{ku\,\atanh r}{r^3}+\frac{k^2t}{r^2(1-r^2)},
\end{align*}
where we used $\frac{d}{dr}\frac{\atanh r}{r}=\frac{1}{r(1-r^2)}-\frac{\atanh r}{r^2}$
and $r^2-kt=u$.
\end{proof}

\begin{lemma}\Label{lem:convexity_criterion}
If $k\,\atanh a\ge a^3$ and $k^2\ge u^3$, then $\DeltaBE''(z)\ge0$ for
$|z|<1$, and hence $\DeltaBE$ is convex on $[-1,1]$.
\end{lemma}

\begin{proof}
The function $\atanh(r)/r=\sum_{j\ge0}r^{2j}/(2j+1)$ is increasing on
$(0,1)$, and $r\ge a$ because $r^2=u+kt\ge u$.  Hence
$\atanh(r)/r\ge\atanh(a)/a$, and with the first hypothesis
\[
\frac{ku\,\atanh r}{r^3}=\frac{ku}{r^2}\cdot\frac{\atanh r}{r}\ \ge\ \frac{ku}{r^2}\cdot\frac{\atanh a}{a}\ \ge\ \frac{u\,a^2}{r^2}=\frac{u^2}{r^2}.
\]
Inserting this into \eqref{eq:second_derivative} and using
$1-r^2=k(1-t)$, so that $\frac{k^2t}{r^2(1-r^2)}=\frac{kt}{r^2(1-t)}$, we get
\[
\DeltaBE''(z)\ \ge\ -\frac{u}{1-ut}+\frac{u^2}{u+kt}+\frac{kt}{(u+kt)(1-t)}
=\frac{t\,[\,k^2-u^3+u^3t\,]}{(u+kt)(1-t)(1-ut)},
\]
where the last equality is a direct expansion over the common denominator.
All denominators are positive for $|z|<1$, and the numerator is
nonnegative when $k^2\ge u^3$.  Continuity extends convexity to the closed
interval.
\end{proof}

\begin{proof}[Proof of Theorem~\ref{thm:witness}, single-use part]
At $p=621/5000$ one has $a=1879/2500$, and the two hypotheses of
Lemma~\ref{lem:convexity_criterion} follow from finite rational
computations.  Since the Taylor series of $\atanh$ has positive
coefficients,
\begin{gather*}
\atanh a\ >\ \sum_{j=0}^{8}\frac{a^{2j+1}}{2j+1}\ >\ \frac{122}{125},\qquad
k\cdot\frac{122}{125}-a^3>\frac{7}{100000},\\
k^2-u^3>\frac{9}{1000},
\end{gather*}
so $k\atanh a>a^3$ and $k^2>u^3$.  Hence $\DeltaBE$ is convex in $x$ on
$[0,1]$, its lower convex envelope is $\DeltaBE$ itself, and
Lemma~\ref{lem:binary_single_use} gives $C^{(1)}=0$.  Equivalently, for
every auxiliary-variable construction Jensen's inequality gives
$I(V;B)-I(V;E)=\DeltaBE(P_X)-\sum_vP_V(v)\DeltaBE(P_{X|V=v})\le0$, and a
constant $V$ attains zero.
\end{proof}

The criterion is consistent with the single-use analysis of \cite{TBR}
and establishes the converse here by finite rational inequalities.

\subsection{Certified positivity of the two-use rate}

\begin{proof}[Proof of Theorem~\ref{thm:witness}, two-use part]
We bound the two-use rate by enclosing the output entropies in rational
intervals. The convolution is
$P_2=P_3*P_1=(0.217,\,0.283,\,0.283,\,0.217)$.
Bob's output distributions for $P_1$ and $P_2$ are rational, since they
are convolutions with two independent $\Bern(p)$ flips. Their entropies
are therefore sums of terms $-y\log y$ with rational $y$.

For Eve, the output state is $DM(P)D$ as in
\eqref{eq:eve_two_copy_structure}, where
$D=\operatorname{diag}(\sqrt{w_j})$ and $M(P)_{jk}=c_P(j\oplus k)$.
This state is similar to the rational matrix
$M(P)\operatorname{diag}(w)$. Write
$\chi_i(t):=\det[tI-M(P_i)\operatorname{diag}(w)]$, and set
\[
\ell_1=(0.014,\,0.094,\,0.117,\,0.774),\qquad
\ell_2=(0.015,\,0.094,\,0.123,\,0.767).
\]
For each $i=1,2$ and $j=1,\ldots,4$, exact substitution shows that
$\chi_i$ has sign $(-1)^{j-1}$ at $\ell_{i,j}$ and the opposite sign at
$\ell_{i,j}+10^{-3}$. These four disjoint intervals therefore isolate
all four eigenvalues.
Seventy bisections reduce each width to $10^{-3}2^{-70}<10^{-24}$.

To enclose the logarithms, write each argument as $2^m\varrho$ with
$1\le\varrho\le2$, and use
\[
\log\varrho=2\sum_{j=0}^{29}\frac{s^{2j+1}}{2j+1}+\mathcal R,
\qquad s=\frac{\varrho-1}{\varrho+1},\qquad
0\le\mathcal R\le\frac{2s^{61}}{61(1-s^2)}.
\]
The same formula bounds $\log2$. Interval multiplication then bounds
$-y\log y$, without requiring it to be monotone on the interval.
Exact rational evaluation gives the enclosures
\begin{align*}
0.6360630&<\DeltaBE_2(P_1)<0.6360631,\\
0.6360651&<\DeltaBE_2(P_2)<0.6360652.
\end{align*}
Taking half their difference proves \eqref{eq:witness_interval}.
\end{proof}

\section{From BPSK to the hybrid binary channel}\Label{sec:bpsk}

The hybrid binary channel arises from quantum BPSK when Bob uses a
fixed binary measurement while Eve keeps her quantum output. We first
identify the two induced channels and then give a beam-splitter realization.

\subsection{Eve's channel}\Label{subsec:bpsk_to_eve}
Keep $\calH_E$, $\rho_E=\ket{\alpha'}\!\bra{\alpha'}$ and the parity
representation $U_E$ of Section~\ref{subsec:qbpsk}, and for $\alpha'\ne0$ define the even and
odd cat-type states
\begin{align}
\ket{\mathrm{even}'}&:=(\cosh|\alpha'|^2)^{-1/2}\sum_{m\ge0}\frac{\alpha'^{\,2m}}{\sqrt{(2m)!}}\ket{2m},\\
\ket{\mathrm{odd}'}&:=(\sinh|\alpha'|^2)^{-1/2}\sum_{m\ge0}\frac{\alpha'^{\,2m+1}}{\sqrt{(2m+1)!}}\ket{2m+1}.
\end{align}
By \cite[Eq.~(12)]{HV},
\begin{equation}
\ket{\pm\alpha'}=\sqrt{a_{e|\alpha'}}\,\ket{\mathrm{even}'}\pm\sqrt{a_{o|\alpha'}}\,\ket{\mathrm{odd}'},
\qquad
a_{e|\alpha}:=e^{-|\alpha|^2}\cosh|\alpha|^2,\quad a_{o|\alpha}:=e^{-|\alpha|^2}\sinh|\alpha|^2.
\end{equation}
The two states therefore live in the two-dimensional space spanned by
$\ket{\mathrm{even}'}$ and $\ket{\mathrm{odd}'}$, on which the parity acts as
$\operatorname{diag}(1,-1)$.  Setting $q':=a_{o|\alpha'}<\tfrac12$
identifies Eve's channel with the one of Section~\ref{subsec:hybrid_binary}
after exchanging the even and odd basis vectors; the resulting sign change
of the parity operator does not affect its conjugation action.
This is also the invariant subspace referred to after Lemma~\ref{L3}.
At zero amplitude the two signal states coincide and the odd-cat
normalization above is not used; this degenerate case is treated directly
as a constant-output channel.

\subsection{Bob's channel}\Label{subsec:bpsk_to_bob}
For $\alpha\ne0$, define $\ket{\mathrm{even}}$ and $\ket{\mathrm{odd}}$ in the same way from
$\alpha$ and let $\ket\pm:=(\ket{\mathrm{even}}\pm\ket{\mathrm{odd}})/\sqrt2$.
For equal prior probabilities of the two symbols, the projective
measurement $\{\ket+\!\bra+,\ket-\!\bra-\}$ is the Helstrom measurement that
minimizes the error probability of distinguishing $\ket\alpha$ from
$\ket{-\alpha}$ \cite{Helstrom}.  Its success probability is
\begin{equation}
|\langle\alpha|+\rangle|^2
=\frac12\Bigl(\sqrt{a_{e|\alpha}}+\sqrt{a_{o|\alpha}}\Bigr)^2
=\frac12\Bigl(1+2\sqrt{a_{e|\alpha}a_{o|\alpha}}\Bigr)
=\frac12\Bigl(1+\frac{\sqrt{2\sinh2|\alpha|^2}}{e^{|\alpha|^2}}\Bigr).
\end{equation}
If Bob fixes this measurement, the channel from Alice to him is a binary
symmetric channel with crossover probability
\begin{equation}
q=1-|\langle\alpha|+\rangle|^2=\frac12\Bigl(1-\frac{\sqrt{2\sinh2|\alpha|^2}}{e^{|\alpha|^2}}\Bigr),
\end{equation}
which is Bob's channel in Section~\ref{subsec:hybrid_binary}.  We emphasize
that this receiver-measurement-induced model is not the unrestricted BPSK
wiretap channel, in which Bob keeps his quantum output and may measure
collectively across channel uses; it is the cq channel seen by a receiver
with a fixed symbol-by-symbol detector. If $\alpha=0$, Bob's
two states coincide and the induced binary channel is treated directly,
without defining an odd-cat vector.

\subsection{A physical realization}\Label{subsec:physical}
Let Alice's system be one bosonic mode and consider a lossless beam
splitter network with three output modes, Bob's $\calH_B$, Eve's
$\calH_E$ and an environment $\calH_R$, acting on coherent states by the
isometry
\begin{equation}
V\ket\beta_A=\ket{a_B\beta}_B\ket{a_E\beta}_E\ket{a_R\beta}_R,\qquad |a_B|^2+|a_E|^2+|a_R|^2=1.
\end{equation}
With $U_{A,0}=I$ and $U_{A,1}=\sum_n(-1)^n\ket n\!\bra n$ on Alice's mode,
\begin{equation}
V\bigl(U_{A,1}\ket\beta_A\bigr)=V\ket{-\beta}_A
=\bigl(U_{B,1}\ket{a_B\beta}_B\bigr)\bigl(U_{E,1}\ket{a_E\beta}_E\bigr)\ket{-a_R\beta}_R,
\end{equation}
so with $\alpha:=a_B\beta$ and $\alpha':=a_E\beta$ the pair of channels to
$B$ and $E$ is the quantum BPSK example, provided the environment mode $R$
is inaccessible to Eve (or $a_R=0$).  If Eve could also collect $R$, her
two states would be $\ket{\pm a_E\beta}\ket{\pm a_R\beta}$, with effective
energy $(|a_E|^2+|a_R|^2)|\beta|^2$, and the parameter $q'$ above would have
to be computed from that energy instead.

\section{Conclusion}\Label{sec:conclusion}

Covariance gives a direct expression for the private information of a
translated-noise encoder through group convolution. The construction
applies at every finite blocklength and allows correlations in both
distributions. Constant-composition coding and balanced hashing achieve
these rates for arbitrary outer distributions; uniform randomization
allows a further reduction through Eve's representation.

For the hybrid binary channel, optimizing the sender's noise beyond
independent physical flips and a symmetric logical flip improves our numerical estimate of the optimized
repetition benchmark at three and four uses. The explicit
two-use witness also gives a rigorous separation between the single-use
private information and the private capacity. The optical realization
identifies a receiver setting in which this hybrid model arises.

Further questions include exactness of the covariant bound for additional
channel classes, the behavior of correlated noise at larger blocklengths
and for groups such as $\bbZ_d$, and refinements of the secrecy bound at
finite coding blocklength.

\begin{acknowledgments}
M.H. acknowledges support in part from the Guangdong Provincial Quantum
Science Strategic Initiative (Grant No.~GDZX2505003) and the General
R\&D Projects of the 1+1+1 CUHK--CUHK(SZ)--GDST Joint Collaboration Fund
(Grant No.~GRDP2025-022). F.S. acknowledges support from the European
Commission through a Marie Sk{\l}odowska-Curie Global Fellowship. The core ideas in this paper were developed by the authors well before the widespread use of AI in the research community. AI was utilized to assist with the writing, numerical analyses, manuscript preparation, and other minor technical tasks.
\end{acknowledgments}

\appendix

\section{Choi-operator form of covariance}\Label{app:choi}

For a quantum channel $\calN:\calB(\mathsf A)\to\calB(\mathsf B)$,
covariance means
$\calN(U_g\rho U_g^\dagger)=V_g\calN(\rho)V_g^\dagger$.
Fix $\ket\Phi=d^{-1/2}\sum_i\ket i\ket i$ and the normalized Choi operator
$J(\calN)=(\mathrm{id}\otimes\calN)(\ket\Phi\!\bra\Phi)$.
The identity $(I\otimes X)\ket\Phi=(X^T\otimes I)\ket\Phi$ gives
\begin{equation}
J(\calN)=(U_g^T\otimes V_g^\dagger)J(\calN)
                    (U_g^T\otimes V_g^\dagger)^\dagger.
\Label{eq:choi_covariant_conj}
\end{equation}
Equivalently, $J(\calN)$ commutes with $U_g^*\otimes V_g$ for every $g$.
Changing the basis conjugates this condition, and exchanging the Choi
factors exchanges the representation factors. For covariance under the
full unitary group with $V_g=U_g$, the commutant is spanned by $I$ and
$\ket\Phi\!\bra\Phi$, giving the isotropic Choi family of depolarizing
channels. For cq channels, the input action is the permutation action of
Section~\ref{subsec:cq_covariance}.

\section{Numerical methods}\Label{app:numerics}

\subsection{Grid, conventions and verification}
All computations of Section~\ref{sec:TBR} use natural logarithms and the
exact grid $p_k=(5929+4k)/49000$, $k=0,\dots,49$. Distributions on
$\bbZ_2^n$ are indexed in binary-integer order, most significant bit first;
for $n=2$ the order is $(00,01,10,11)$. Each optimized pair is normalized
before evaluating its rate. Computations use double precision, with
$0\log0=0$ in the entropy calculation; negative eigenvalues within
roundoff are set to zero, and larger violations are rejected.

We checked all $150$ achieved rates by constructing the physical
tensor-product output states directly, independently of the exchangeable
reduction. The largest observed discrepancy was
$1.2\times10^{-15}$ nats per use. This numerical consistency check is
separate from the exact rational bounds in Section~\ref{sec:witness}
and the moment determinant in \eqref{eq:four_use_profile}.

\subsection{Search and benchmark evaluation}
At each grid point and each $n\in\{2,3,4\}$ we begin with the optimized
logically prefixed repetition code. Exchangeable noise is then optimized
by analytic-gradient SLSQP, starting from this code, the preceding grid
point's solution and eight random Dirichlet profiles. The best candidate
is refined over both full probability vectors $(P_1,P_3)$ by SLSQP.
For $n=2$, we also use differential evolution over eight positive raw
parameters, normalized separately to obtain $P_1$ and $P_3$, with
population multiplier $8$ and at most $120$ generations, followed by
full-pair local refinement. Improvements no larger than
$2\times10^{-13}$ nats per use are treated as ties, in which case the
simpler repetition candidate is retained. These searches give feasible
pairs and achievable rates, without a global optimality assertion.

For $R_n^{\rep,0}$ we scan $1001$ values of $\lambda\in[0,1/2]$ and
refine the best by bounded scalar optimization. For
$R_n^{\rep,\log}$ we use analytic-gradient L-BFGS-B from $42$ starting
points and the preceding grid point's solution. To evaluate the full
benchmark $R_n^{\rep}$, we sample the binary-prior function in
\eqref{eq:R_rep_envelope}, compute its lower convex hull, and maximize
the gap. The physical flip probability is scanned at $51$ points over
$[0,1/2]$, with $201$ prior samples; the optimized logical-flip parameter
is also included. Bounded scalar refinement uses $801$ prior samples,
and the final evaluation uses $2001$. At every positive point the
maximizing prior is $x=1/2$, and the full benchmark agrees with
$R_n^{\rep,\log}$ to within $7.1\times10^{-9}$ nats per use. The
figures and tables therefore use the latter as the numerical estimate
of $R_n^{\rep}$.

\end{document}